\documentclass[%
    aip,
    jmp,
    amsmath,
    amssymb,
    preprint,
]{revtex4-1}

\usepackage{graphicx}

\usepackage[utf8]{inputenc}
\usepackage[T1]{fontenc}

\usepackage{mathptmx}

\usepackage{amsthm} 

    \theoremstyle{plain}
        \newtheorem{theorem}{Theorem}
        \newtheorem{proposition}{Proposition}

    \theoremstyle{definition}
        \newtheorem{definition}{Definition}

    \theoremstyle{remark}
        \newtheorem{remark}{Remark}

\begin{document}

\title{Resurgence and holonomy of the $\phi^{2k}$ model in zero dimension}

\author{Fr\'ed\'eric Fauvet}
    \affiliation{IRMA, Universit{\'e} de Strasbourg, 7 rue Descartes,
    67084 Strasbourg Cedex, France.}
    \email{fauvet@math.unistra.fr}
\author{Fr\'ed\'eric Menous}
    \affiliation{Universit\'e Paris-Saclay, CNRS,
    Laboratoire de math\'ematiques d'Orsay, 91405, Orsay, France.}
    \email{frederic.menous@universite-paris-saclay.fr}
\author{Julien Qu\'eva}
    \affiliation{Universit\'e de Corse -- CNRS UMR 6134 SPE,
    Campus Grimaldi BP 52, 20250 Corte, France.}
    \email{queva@univ-corse.fr}

\date{\today}

\begin{abstract}
    We describe the resurgence properties of some partition functions
    corresponding to field theories in dimension 0. We show that these functions
    satisfy linear differential equations with polynomial coefficients and then
    use elementary stability results for holonomic functions to prove resurgence
    properties, enhancing previously known results on growth estimates for the
    formal series involved, which had been obtained through a delicate
    combinatorics.
\end{abstract}

\maketitle


\section{Introduction}

The theory of resurgent functions and alien differential calculus has been
introduced in the late 70's,  and then developed single--handedly by Jean 
Ecalle,
in the field of singularities of dynamical systems with complex analytic data.
It was a major discovery that the divergent series appearing in the formal
solutions at a singular point, when expressed as expansions in a suitable
variable $z$ close to $\infty$, admit Borel transforms with isolated
singularities which can be analyzed with operators of a
new kind: the alien derivations.

Resurgent functions and alien calculus have made it possible to solve difficult 
problems of classification of dynamical systems and to tackle singular 
perturbation theory, in particular for the Schr{\"o}dinger equation, which 
display features of resurgence with respect to $\hbar$. Recently alien calculus 
has been applied in a number of domains of Theoretical Physics: Quantum 
Mechanics, perturbative
Quantum Field Theory, matrix models, topological strings, etc 
{\cite{ANICETO2019}}.

In the present text, we focus on the \emph{integrals} discussed by
Rivasseau et al \cite{rivasseau_loop_2010, lionni_note_2018, Rivasseau2018}
that correspond to scalar quantum field theories in dimension 0. That is, we
consider scalar fields $\phi \in \mathbb{R}$ and a Lagrangian
\[
\mathcal{L}_k (\phi) = - \frac{1}{2} \phi^2 - \lambda \phi^{^{2 k}},\quad 
(k \geqslant 2)
\]
with ``potential'' $V (\phi) = \lambda \phi^{2 k}$ and study the associated
``partition function'':
\[ Z_0^{} (\lambda) = \int_{\mathbb{R}} e^{- \frac{1}{2} \phi^2 - \lambda
    \phi^{2 k}} \frac{d \phi}{\sqrt{2 \pi}}. \]
We shall consider its related moments
\[ Z_{2 j} (\lambda) = \int_{\mathbb{R}} \phi^{2 j} e^{- \frac{1}{2} \phi^2
    - \lambda \phi^{2 k}} \frac{d \phi}{\sqrt{2 \pi}}, \quad (j \geqslant 1)
\]
and the ``Free energy'' $W (\lambda) = \log Z_0 (\lambda)$.

Classical results, reminded in section \ref{ClassicIntegrals} for a wider
class of potentials, that is polynomials $V \in \mathbb{R}_{2 k} [\phi]$ of
degree $2 k$ with positive dominant coefficient, already allow to prove the
following:

\begin{proposition}
    For a fixed $k \geqslant 1$, the integrals $Z_{2 j} (\lambda)$ are
    convergent for $\Re \lambda \geqslant 0$ and define continuous functions on
    the half-plane $\overline{S} = \left\{ \lambda \in \mathbb{C} ;\ \Re \lambda
    \geqslant 0 \right\}$ that are
    analytic in $S = \left\{ \lambda \in \mathbb{C};\ \Re \lambda > 0 \right\}$.
    Morevover these integrals have asymptotic
    expansions in $\overline{S}$: there exists $(a_n^j)_{n \geqslant 0}$ such 
    that,
    for all $N \geqslant 0$,
    \[ \lim_{\substack{\lambda \rightarrow 0\\\lambda \in S}}\ 
    \lambda^{- N} \left[ Z_{2 j} (\lambda) - \sum_{n = 0}^N
    a^j_n \lambda^n_{} \right] = 0. \]
\end{proposition}

This is the classical definition of an asymptotic expansion, following
Ref.~\onlinecite{wasow} we note $\widetilde{Z_{2j}}$ the formal series $\sum_{n 
=
    0}^{\infty} a^j_n \lambda^n_{}$ that corresponds to the ``perturbative
expansion'' of $Z_{2 j} (\lambda)$. These are results already obtained in
Refs.~\onlinecite{rivasseau_loop_2010, lionni_note_2018} but it happens that
the partition function, as well as its asymptotic expansion, are solutions of
a differential equation, so that resurgence theory for the case of linear
differential equations will in the end provide more precise results on the
partition function and its logarithm.

As we shall see in section \ref{Ek}, we have the following:

\begin{theorem}
    \label{th:Ek}For the potential \ $V (\phi) = \phi^{2 k}$, $k \geqslant 1$,
    the partition function $Z_0$, as well as its asymptotic expansion, do
    satisfy the differential equation $(E_k)$
    \begin{equation}
    \left[ \left( \prod_{j = 0}^{k - 1} (2 k \lambda \partial_{\lambda} + 2 j
    + 1) \right) + \partial_{\lambda}  \right].Z_0 = 0.
    \end{equation}
\end{theorem}

This equation completely determines the resurgence of the partition function and
consequently also the resurgence properties of the free energy $W = \log Z_0$.
The \emph{parameter}
$\lambda$ now plays the role of a \emph{variable} in a linear ordinary
differential equation with polynomial coefficients and we can make use of the
general theory of holonomic functions to derive results for the function
$\widetilde{Z_0} (\lambda)$.

Thus, we shall see below how to obtain without any combinatorics the
Gevrey--1 growth for the coefficients of the series $\widetilde{Z_0}
(\lambda)$, but also beyond that analyze the singularities of its Borel
transform and reach resurgence properties.

\begin{remark}
    The derivation of linear ODEs for some families of integrals depending on
    parameters is of course the topic of Picard--Fuchs theory, with its
    far reaching generalizations (Gauss--Manin connection); the present text 
    implements essentially elementary techniques, focussing on the resurgence 
    properties, for the family of integrals under consideration.
\end{remark}

We remind in section \ref{ClassicIntegrals} some analytic properties of
partition functions in dimension 0 that will be used in section \ref{Ek},
devoted to the different recursive relations between the partition function
and its moments so that we can give the proof of theorem \ref{th:Ek}. Sections
\ref{s:sing} and \ref{s:res} recall basic results on Borel--Laplace summation
of divergent series and resurgence theory respectively. In sections
\ref{s:resEk} and \ref{s:nonlin} we perform an analysis of the equations $(E_k)$
with the tools of alien calculus and then investigate the resurgence
properties of the free energy functions. In a last section we show how the
same techniques might apply to treat some particular cases of singular
perturbation problems for the Schr{\"o}dinger equation and we implement in
detail the method for the easy ``Airy case''.

\noindent\textbf{Data Availability Statement:} The data that supports the
findings of this study are available within the article.
 
\section{Properties of partition functions: estimates}\label{ClassicIntegrals}

We give here some general results for the partition function for a potential
$V (\phi) = \sum_{i = 0}^{2 k} v_i \phi^i$ with $k \geqslant 1$ and $v_{2 k} >
0$ and its related moments, noticing that all these functions share the same
integral shape
\[ \int_{\mathbb{R}} P (\phi) e^{- \frac{1}{2} \phi^2} e^{- \lambda V
    (\phi)} \frac{d \phi}{\sqrt{2 \pi}}, \]
where $P$ is a polynomial.

\subsection{Analytic properties}

\begin{theorem}
    \label{conv}Let $V (\phi) = \sum_{i = 0}^{2 k} v_i \phi^i$ with $k \geqslant
    1$ and $v_{2 k} > 0$ and $P \in \mathbb{R} [\phi]$, the integral
    \[ \int_{\mathbb{R}} P (\phi) e^{- \frac{1}{2} \phi^2} e^{- \lambda V
        (\phi)} \frac{d \phi}{\sqrt{2 \pi}} \]
    is convergent for any $\lambda$ in the closed sector $\overline{S_{}} =
    \left\{ \lambda \in \mathbb{C} ;\ \Re \lambda \geqslant 0 \right\}$
    and defines a function $f_{P, V}$ continuous on
    $\overline{S}$ and analytic in
    $S = \left\{ \lambda \in \mathbb{C};\ \Re \lambda > 0 \right\}$.
    Moreover, for any $n \in \mathbb{N}$ and $\lambda \in S$,
    \[ f_{P, V}^{(n)} (\lambda) = \int_{\mathbb{R}} P (\phi) (- V (\phi))^n
    e^{- \frac{1}{2} \phi^2} e^{- \lambda V (\phi)} \frac{d \phi}{\sqrt{2
            \pi}} . \]
\end{theorem}

\begin{proof}
    The proof is based on variations on Lebesgue's dominated convergence
    theorem. We can first observe that $V$ has a minimum $m$ on $\mathbb{R}$ so
    that:
    \[ \forall \phi \in \mathbb{R},\ \forall \lambda \in \overline{S},
    \quad \left| P (\phi) e^{- \frac{1}{2} \phi^2} e^{- \lambda V (\phi)}
    \right| \leqslant | P (\phi) | e^{- \Re \lambda .m} e^{- \frac{1}{2}
        \phi^2} \]
    and $f_{P, V}$ is well defined. In order to prove the continuity and the
    analyticity of $f_{P, V}$, since $f_{P, V} = e^{- \lambda m} f_{P, V - m}$,
    we can assume without loss of generality that $V$ is non negative on the
    real axis.
    
    Following Ref.~\onlinecite[chapter XIII, section 8]{Dieud}, the map
    \[ \begin{array}{ccccl}
    F & : & \mathbb{R} \times \overline{S} & \rightarrow & \mathbb{C}\\
    &  & (\phi, \lambda) & \mapsto & P (\phi) e^{- \frac{1}{2} \phi^2}
    e^{- \lambda V (\phi)}
    \end{array} \]
    is such that:
    \begin{enumerate}
        \item For any $\lambda \in \overline{S}$, the function
        $\phi \mapsto F (\phi, \lambda)$ is integrable on $\mathbb{R}$.
        
        \item For any $\phi \in \mathbb{R}$, the function $\lambda \mapsto F
        (\phi, \lambda)$ is continuous on $\overline{S}$ and analytic on $S$.
        
        \item For any $\lambda \in \overline{S}$, and any $\phi \in \mathbb{R}$
        \[ \left| P (\phi) e^{- \frac{1}{2} \phi^2} e^{- \lambda V (\phi)} 
        \right|
        \leqslant | P (\phi) | e^{- \frac{1}{2} \phi^2} = g (\phi), \]
        where $g$ is integrable on $\mathbb{R}$.
    \end{enumerate}
    Thanks to Lebesgue's dominated convergence theorem, it automatically ensures
    that $f_{P, V}$ is continuous on $\overline{S}$, analytic on $S$ with the
    attempted formulas for its derivatives.
\end{proof}

One can also observe that $f_{P, V}$ is indeed $\mathcal{C}^{\infty}$ on $[0,
+ \infty [$ and its derivatives at $\lambda = 0$ are given, for $n \geqslant
0$, by
\[
f_{P, V}^{(n)} (0) = \int_{\mathbb{R}} P (\phi) (- V (\phi))^n e^{-
    \frac{1}{2} \phi^2} \frac{d \phi}{\sqrt{2 \pi}}, \]
that are linear combination of gaussian moments:
\[ a_j = \int_{\mathbb{R}} \phi^j e^{- \frac{1}{2} \phi^2} \frac{d
    \phi}{\sqrt{2 \pi}} = \text{} \left\{ \begin{array}{ll}
0 & \text{if $j$ is odd}\\
\frac{j!}{2^{j / 2} (j / 2) !} & \text{if $j$ is even}
\end{array} \right. . \]

This suggests that $f_{P, V}$ has an asymptotic expansion when $\lambda
\rightarrow 0$ in $S$ (see for example Ref.~{\onlinecite[chapter III]{wasow}}).

\subsection{Asymptotics expansions}

\begin{theorem}
    \label{DA}Let $V (\phi) = \sum_{i = 0}^{2 k} v_i \phi^i$ with $k \geqslant
    1$ and $v_{2 k} > 0$, $P \in \mathbb{R} [\phi]$, and for $\Re \lambda
    \geqslant 0$ and $n \geqslant 0$
    \begin{align*}
    f_{P, V} (\lambda) &= \int_{\mathbb{R}} P (\phi) e^{- \frac{1}{2} \phi^2}
    e^{- \lambda V (\phi)} \frac{d \phi}{\sqrt{2 \pi}}, \\
    \alpha_n = \frac{f_{P, V}^{(n)} (0)}{n!}
    &= \int_{\mathbb{R}} \frac{1}{n!} P (\phi)
    (- V (\phi))^n e^{- \frac{1}{2} \phi^2} \frac{d \phi}{\sqrt{2 \pi}} .
    \end{align*}
    For any positive integer N
    \[ \lim_{\substack{\lambda \rightarrow 0\\ \lambda \in S}}\
    \lambda^{- N} \left[ f_{P, V} (\lambda) - \sum_{n = 0}^N
    \alpha_n \lambda^n_{} \right] = 0. \]
\end{theorem}

We note $\widetilde{f}_{P, V} (\lambda)$ the formal series $\sum_{n = 
0}^{\infty}
\alpha_n \lambda^n_{}$ which is the asymptotic expansion of $f_{P, V}$ when
$\lambda \in S$, $\lambda \rightarrow 0$ and write
\[  f_{P, V} (\lambda) \thicksim_{} \sum_{n = 0}^{\infty} \alpha_n \lambda^n, 
\quad \lambda \in S,
\quad \lambda \rightarrow 0. \]
Since, for $n \geqslant 0$,
\[ f_{P, V}^{(n)} (\lambda) = \int_{\mathbb{R}} P (\phi) (- V (\phi))^n e^{-
    \frac{1}{2} \phi^2} e^{- \lambda V (\phi)} \frac{d \phi}{\sqrt{2 \pi}}, \]
it is a matter of fact to check that
\[ f_{P, V}^{(n)} (\lambda) \thicksim_{} \widetilde{f}_{P, V}^{(n)} (\lambda),
\quad \lambda \in S,
\quad \lambda \rightarrow 0, \]
where $\widetilde{f}_{P, V}^{(n)}$ is the $n$-th formal derivative of the formal
series $\widetilde{f}_{P, V} (\lambda)$.

\begin{proof}
    Since the asymptotic expansion of a product of functions is the product of
    their asymptotic expansions\cite{wasow}, we can assume as in the
    previous proof that $V$ is non negative on $\mathbb{R}$. Otherwise, if the
    minimum $m$ of $V$ is negative, we write $f_{P, V} = e^{- \lambda m} f_{P, V
        - m}$, where $e^{- \lambda m}$ has an obvious asympotic expansion as an
    entire function of $\lambda$, this last result can be obtained with the help
    of the Taylor formula that will be useful to achieve this proof:
    \[ \forall N \geqslant 0,\quad \forall \zeta \in \mathbb{C}, \quad
    R_N (\zeta) = e^{\zeta} - \sum_{n = 0}^N \frac{\zeta^n}{n!} = \zeta^{N +
        1} \int_0^1 \frac{(1 - t)^N}{N!} e^{t \zeta} d t \]
    that leads to
    \[ \forall N \geqslant 0, \quad \forall \zeta \in \mathbb{C}, \quad
    | R_N (\zeta) | = \left| e^{\zeta} - \sum_{n = 0}^N \frac{\zeta^n}{n!}
    \right| \leqslant \frac{| \zeta |^{N + 1}}{(N + 1) !} \max \{ 1, e^{\Re
        \zeta} \} . \]
    For $\lambda \in S$, we have
    \[ Q_N (\lambda) = \lambda^{- N} \left[ f_{P, V} (\lambda) - \sum_{n = 0}^N
    \alpha_n \lambda^n_{} \right] = \lambda^{- N} \int_{\mathbb{R}} P (\phi)
    R_N (- \lambda V (\phi)) e^{- \frac{1}{2} \phi^2} \frac{d \phi}{\sqrt{2
            \pi}} \]
    thus
    \[ | Q_N (\lambda) | \leqslant \frac{| \lambda |}{(N + 1) !}
    \int_{\mathbb{R}} | P (\phi) .V (\phi)^{N + 1} | e^{- \frac{1}{2}
        \phi^2} \frac{d \phi}{\sqrt{2 \pi}} \]
    and this ends the proof of theorem.
\end{proof}

For such potentials $V$, the asymptotic expansion $\widetilde{f}_{1, V}$ is, in
physicists' terms, the perturbative expansion of the partition function $f_{1,
    V}$. A crucial remark for the sequel is that the asymptotic expansion of a
function is unique so that, if we have functions satisfying a linear recursive
equation or a linear differential equation, so do the formal series
corresponding to their respective asymptotic expansions (see
Ref.~{\onlinecite[chapter III]{wasow}}, for more properties on asymptotic
expansions).
 
\section{Governing equations for partition functions}\label{Ek}

\subsection{Recursive and differential relations for the moments}

Let $V (\phi) = \sum_{i = 0}^{2 k} v_i \phi^i$ with $k \geqslant 1$ and $v_{2
    k} > 0$ a given potential and its relative moments,
\[ Z_j (\lambda) = \int_{\mathbb{R}} \phi^j e^{- \frac{1}{2} \phi^2} e^{-
    \lambda V (\phi)} \frac{d \phi}{\sqrt{2 \pi}} \]
whose asymptotics expansions are noted $\widetilde{Z}_j (\lambda)$. The key
properties of these moments are the following:

\begin{proposition}
    \label{rec}The family of functions $(Z_j)_{j \geqslant 0}$ as well as the
    family of formal series $(\widetilde{Z}_j)_{j \geqslant 0}$ satisfy the
    following relations:
    \begin{equation}
    \forall j \geqslant 0, \quad Z'_j = - \sum_{i = 0}^{2 k} v_i Z_{j + i}, 
    \label{rec1}
    \end{equation}
    \begin{equation}
    \forall j \geqslant 0, \quad (j + 1) Z_j = Z_{j + 2} + \lambda \sum_{i =
        1}^{2 k} i v_i Z_{i + j}. \label{rec2}
    \end{equation}
\end{proposition}

We call these relations the governing equations associated to the potential
$V$.

\begin{proof}
    The proof is quite elementary: the first equation is a simple
    differentiation under the integral whereas the second one is a simple
    integration by part. For all $j \geqslant 0$, we have
    \begin{align*}
    (j + 1) Z_j
    & =  \int_{\mathbb{R}} (j + 1) \phi^j e^{- \frac{1}{2}
        \phi^2} e^{- \lambda V (\phi)} \frac{d \phi}{\sqrt{2 \pi}}
    =  \int_{\mathbb{R}} (\phi^{j + 1})' e^{- \frac{1}{2} \phi^2} e^{-
        \lambda V (\phi)} \frac{d \phi}{\sqrt{2 \pi}} \\
    & = \left[ \phi^{j + 1} e^{- \frac{1}{2} \phi^2} e^{- \lambda V
        (\phi)} \right]_{- \infty}^{+ \infty}
    +  \int_{\mathbb{R}} \phi^{j + 1} (\phi + \lambda V' (\phi))
    e^{- \frac{1}{2} \phi^2} e^{- \lambda V
        (\phi)} \frac{d \phi}{\sqrt{2 \pi}} \\
    & =  Z_{j + 2} + \lambda \sum_{i = 1}^{2 k} i v_i Z_{i + j}
    \end{align*}
    The uniqueness of asymptotic expansion ensures that the formal series
    $(\widetilde{Z}_j)_{j \geqslant 0}$ have the same properties.
\end{proof}

Let us focus now on the case $V (\phi) = \phi^{2 k}$ for which these recursive
relations will provide us a fundamental differential equation for the
partition function.

\subsection{The case $V (\phi) = \phi^{2 k}$, $k \geqslant 1$}

In this particular case, for any odd $j$, $Z_j \equiv 0$ so that one can focus
on the integrals
\begin{equation*}
\forall j \in \mathbb{N}, \quad
\forall \lambda \in S, \quad
U_j(\lambda) = \int_{\mathbb{R}} \phi^{2 j} e^{- \frac{1}{2} \phi^2} e^{-
    \lambda \phi^{2 k}} \frac{d \phi}{\sqrt{2 \pi}}
\end{equation*}
and the above recursive relations, \eqref{rec1} and \eqref{rec2}, read:
\begin{equation}
\label{eq:SetEq}
\forall j \geqslant 0, \quad \left\{
\begin{aligned}
\partial_{\lambda} U_j & =  - U_{j + k}\\
(2 j + 1) U_j & = U_{j + 1} + 2 k \lambda U_{j + k}
\end{aligned}
\right.
\end{equation}

In the above set \eqref{eq:SetEq} the second line can be
recognized to be the Schwinger-Dyson equations written on the unnormalized
even green functions with $U_j = Z_{2j} = Z_0 G_{2j}$ (e.g. see
Refs.~\onlinecite{Bender:1988bp,Okopinska:1990pt} where those are addressed by 
truncation
for $k=2$).
The first equation in \eqref{eq:SetEq}, however, is of a kind akin to a
Step-$k$ equation as seen in Ref.~\onlinecite{Argyres:2001sa}.

Combining both equations from \eqref{eq:SetEq}, namely plugging the right hand
side of the first into the second, brings about 
\begin{equation*}
\forall j \geqslant 0, \quad
(2 k \lambda \partial_{\lambda} + 2 j + 1) U_j = U_{j + 1}
\end{equation*}
that leads to

\begin{theorem}
    For the potential \ $V (\phi) = \phi^{2 k}$, $k \geqslant 1$, the partition
    function $U_0$, as well as its asymptotic expansion, do satisfy the
    differential equation $(E_k)$
    \begin{equation*}
    \left[ \left( \prod_{j = 0}^{k - 1} (2 k \lambda \partial_{\lambda} + 2 j
    + 1) \right) + \partial_{\lambda}  \right] .U_0 = 0
    \end{equation*}
\end{theorem}
\begin{proof}
    The proof is straightforward using \eqref{rec1} and \eqref{rec2}:
    \[ \left( \prod_{j = 0}^{k - 1} (2 k \lambda \partial_{\lambda} + 2 j + 1)
    \right) .U_0 = U_k = - \partial_{\lambda} U_0 . \qedhere \]
\end{proof}
This equation $(E_k)$ will provide all the information on the resurgence of
the partition function.

\section{Divergent series and ODEs}\label{s:sing}

\subsection{Singularities of linear ODEs}

In the present article, we shall be concerned with linear ordinary
differential equations with polynomial or rational coefficients; it is in this
context that we now briefly recall some of the basic elements of the theory of
irregular singularities of ODEs with analytic data.

Let us consider such a linear ODE, for a complex variable $x$:
\begin{equation}
a_n (x) y^{(n)} (x) + \ldots + a_1 (x) y' (x) + a_0 (x) y (x) = 0,
\label{e:edolin}
\end{equation}
where the coefficient functions $a_i (x)$ belong to $\mathbb{C} [x]$.

Each point $x_0$ in the complex plane which is not a root of the leading
coefficient $a_n (x)$ is \emph{regular}: at such a point, Cauchy--Lipschitz
theorem applies and we have a basis of the vector space of solutions of 
\eqref{e:edolin}
composed of functions which are analytic near $x_0$.

Moreover, it is the specificity of linear equations that all the local
solutions can be continued along any path $\gamma$ starting at $x_0$ that
avoids the \emph{singular points} of the equation -- namely the roots of
$a_n$ and this global property of the solutions will be of constant use later
on. Singular points of these linear ODEs can be of 2 types:

-- \emph{Regular--singular}, around which a basis of solution will only
involve analytic functions, possibly ramified at the singular point, together
with logarithms, that is, when $x_0=0$, solutions:
\begin{equation*}
y(x) = x^{\beta}\left(\sum_{i=0}^m\Phi_i(x)\log^i x\right),\
\beta\in \mathbb{C},\
\Phi_i \in \mathbb{C}\lbrace x\rbrace.
\end{equation*}

-- \emph{Irregular--singular}, for which, when $x_0 = 0$,  local solutions will 
concomitantly
involve exponentials of so-called determining factors $e^{q (x^{- \nu})}$ (with 
$q (u) \in
\mathbb{C} [u ]$) and divergent series,  that is solutions:
\begin{equation*}
y(x) = x^{\beta} e^{q (x^{- \nu})}\left(\sum_{i=0}^m\Phi_i(x)\log^i x\right),
\ \beta\in \mathbb{C},
\ \Phi_i \in \mathbb{C}[[ x ]],
\end{equation*}
with a growth order of Gevrey--type for the formal series $\Phi_i$ in the 
variable $z=x^{-1}$:

\begin{definition}
    A formal series $\sum_{n \geqslant 0} a_n z^{- n} \in \mathbb{C} [[z^{-
        1}]]$ is said to be Gevrey of order $s$, where $s$ is a positive real 
        number
    (Gevrey--s for short, or of level $q = \frac{1}{s}$) iff:
    \[ \exists C, A > 0, \text{such that}\  \forall n \geqslant 0,\
    | a_n | \leqslant C A^n  (n!)^s. \]
\end{definition}

Formal series with such factorial growth estimates have been
systematically studied by Maurice Gevrey during the 1910' for PDEs. In the
context of linear ODEs a number of results involving these ``Gevrey series''
have notably been obtained by Maillet and Perron; they are ubiquitous in
solutions of linear and non linear differential equations with meromorphic
data, as shown in many works since the seminal work of Ramis\cite{RAM78},
which marked a revival of the topic.

Generically, such Gevrey series that are formal solutions of an ODE at an
irregular--singular point will be divergent, yet there exists a standard
process to express with them analytic solutions \emph{in sectors} at the
singular point: Borel--Laplace summation.

\subsection{Borel, Laplace and Stokes}

The paradigmatic example\cite{LR17} of an ODE with an irregular singularity is 
the
so--called Euler equation, which can be given in the following form
\begin{equation}
x^2 f' (x) = - f (x) + x.
\label{e:euler}
\end{equation}
This equation has a singularity at the origin $x = 0$ and we find there a
unique formal series solution
\[ \widetilde{f_0}  (x) = \sum_{n= 0}^\infty\ (- 1)^n n! x^{n + 1}. \]
This formal series is Gevrey--1 and divergent. The general formal solution of
\eqref{e:euler} is
\[ \Phi (x) = \widetilde{f_0}  (x) + \sigma e^{\frac{1}{x}}, \qquad
(\sigma \in \mathbb{C}). \]
On this simple example we indeed notice the simultaneous presence of Gevrey--1
divergence and of an exponential factor $e^{\frac{1}{x}}$. Equation
\eqref{e:euler} can be given in linear homogeneous form: after division by $x$,
derivation and finally multiplication by $x^2$ we obtain the following linear
second order ODE
\begin{equation}
x^3 f'' (x) + (x^2 + x) f' (x) - f (x) = 0
\label{e:eulerlin} 
\end{equation}
and the vector space of solutions of \eqref{e:eulerlin} is spanned by
$\widetilde{f_0}  (x)$ and $e^{\frac{1}{x}}$.

As a function of the variable $x$ close to the origin, $e^{\frac{1}{x}}$ will
have a very different behaviour according to the direction along which $x
\longrightarrow 0$: it will be explosive when $\Re x$ is positive,
will vanish when it is negative and oscillate when $x$ is purely imaginary.

This completely elementary observation is in fact crucial: for matters of
resummation of divergent series, the situation is \emph{polarized}.
Polarizability is built--in the general solution of the problem and any
sensible summation process will have to take care of this; thus, the
consideration of sectors of opening $\pi$ in the case of Gevrey--1 series is
relevant precisely because of the concomitant $e^{\frac{\alpha}{x}}$ occurence
of exponential factors, for solutions close to the singular point.

We are now ready to introduce the Borel--Laplace mechanism; it is convenient
to change from $x \thicksim 0$ to $z \thicksim \infty$.

\begin{definition}
    Let $f (z) = \sum_{n=0}^\infty\  a_n z^{- (n + 1)}$, we define its Borel
    transform by
    \[ \mathcal{B} (f) (\zeta) : = \sum_{n=0}^\infty\  \frac{a_n}{n!} \zeta^n =
    \hat{f}  (\zeta). \]
\end{definition}

The shift by one unit in the exponents simplifies the formulas and changing
the name of the variable is most useful. For the moment, we thus suppose in
the definition of the Borel transform that the constant term of $f (z)$ is
null; we shall in fact see later how to enhance this transformation to more
general formal series and in the applications we have in view it will always
be possible to deal with series with $a_0 = 0$.

Thus, for the formal series solution of Euler's equation, $\widetilde{f_0} 
(z) = \sum_{n=0}^\infty\  (- 1)^n n! z^{- (n + 1)}$, we have
\begin{equation*}
\mathcal{B} (\widetilde{f_0}) (\zeta)
= \sum_{n=0}^\infty\ (- 1)^n \zeta^n = \frac{1}{1 + \zeta}. 
\end{equation*}
On any direction $d_{\theta} = e^{i \theta} \mathbb{R}_{>}$ in the Borel
plane, except $\mathbb{R}_{<}$, $\widehat{f_0}  (\zeta)$ can be analytically
continued and it decays at $\infty$; thus we can consider its Laplace
transforms on any non--singular direction
\[ \mathcal{L}_{d_{\theta}} \hat{f}  (z)
:= \int_0^{e^{i \theta} \infty}
e^{- z \zeta} \widehat{f_0}  (\zeta) d \zeta.\]
For any $d_{\theta} \neq \mathbb{R}_{<}$ this integral yields an analytic
function in a half plane $\{ \Re (z e^{i \theta}) > \text{constant}
\}$ bisected by $e^{- i \theta}$, which translates in an analytic function in
a sector of opening $\pi$ for $x \thicksim 0$.

As both Borel and Laplace transforms are \emph{morphisms of differential
    algebras}, on the spaces we are considering, it automatically ensures that
the analytic functions obtained by this process are solutions of Euler's
equation: we have performed a \emph{resummation} of the divergent series
$f_0 (x)$ in sectors at the origin\cite{LR94}.

When we move the direction $d_\theta$ without crossing the singular direction,
the Laplace sums coincide on the overlapping sectors (by Cauchy's formula and
Lebesgue dominated convergence at $\infty$) but it is not so when we perform
the integration slightly above and slightly under the singular direction
$\mathbb{R}_{<}$: the 2 integrals differ -- this is \emph{Stokes
    phenomenon}.

In this case, the difference can be explicitely calculated because we have a
closed--form formula for $\widehat{f_0}  (\zeta)$, that has as single
singularity $\omega = - 1$, which is a simple pole with residue equal to $1$
and we get, for $\varepsilon > 0$, $\mathcal{L}_{d_{\pi - \varepsilon}}
\widehat{f_0}  (z) -\mathcal{L}_{d_{\pi + \varepsilon}}  \widehat{f_0}  (z) =
2 \pi i e^z = 2 \pi i e^{\frac{1}{x}}$.

For Euler's equation, everything is explicit but this simple case displays
features which will be quite general for solutions of ODEs with analytic data:
in generic cases, when expressed in a \textbf{suitable variable $z \sim
    \infty$,} these series $\widetilde{f} (z)$ have Borel transforms $\hat{f}$ 
    which
are convergent and a Laplace transform of these $\hat{f}$, when justified,
yields ``sectorial sums'' of the $\widetilde{f} (z)$.

Stokes phenomenon is precisely the fact that, on some overlapping sectors,
these sums may differ.

In practice, the functions $\hat{f} (\zeta)$ have analytic continuations with
\textbf{isolated singularities} $\omega$ and the very presence of
the singularities of $\hat{f}$ accounts for the divergent character of
$\widetilde{f}$; the analysis of these singularities is achieved by alien 
calculus
and Stokes phenomenon can eventually be expressed through the use of alien
derivations on the formal solutions.

\subsection{The Newton polygon}\label{sec:NP}

Newton polygons (NP for short) are pervasive for the study of singularities;
in the context of singularities of linear ordinary differential equations,
they were introduced by J.--P. Ramis in the seminal Ref.~\onlinecite{RAM78} as
a crucial tool for Gevrey asymptotics.

The NP of a linear ODE encodes many properties of the formal solutions of
this equation; we recall now the definition and main properties of the Newton
polygon for a linear differential operator with polynomial coefficients,
see Refs.~\onlinecite{LR17} and \onlinecite{VS03} for further information and
proofs of the classical results we shall use. As such equations arise naturally 
in the sequel, we consider, near $x=0$ linear ODEs such as
\begin{equation}\label{e:edolinr}
H(f)(x)=a_n (x) f^{(n)} (x) + \ldots + a_1 (x) f' (x) + a_0 (x) f (x) = 0
\end{equation}
where the coefficients functions $a_i (x)$ belong to $\mathbb{C} (x)$. Up to a  
polynomial factor, such equations do not differ from equation 
\eqref{e:edolin},  but in this case we will deal  with rational fractions $h\in 
\mathbb{C}(x)$ that can be written as a Laurent series 
\[
h(x)=\sum_{q\geqslant -N} h_{q}x^q.
\] 

\begin{definition}
    We consider such a linear differential equation $H (f) = 0$, where $ H\in
    \mathbb{C} (x) [d / d x]$ and we express such an operator $H$ of order $n$
    in the following way:
    \[ H = H_n (x) \theta^n + \ldots + H_1 (x) \theta + H_0 (x) \quad
    \text{where}\quad \theta = x \frac{d}{d x}. \]
    We consider the set $\mathcal{S}$ of pairs $(i, q)$, with $0 \leqslant i
    \leqslant n$ and $q \in \mathbb{Q}$ such that the monomial $x^q$ appears
    with a non--vanishing coefficient in the Laurent series of $H_i (x)$.
    
    The NP at $0$ of equation \eqref{e:edolinr} is the convex hull of the set 
    $\{ (u, v),
    (i, q) \in \mathcal{S}, 0 \leqslant u \leqslant i, q \leqslant v \}$ in
    $\mathbb{R}^+\times \mathbb{R}$. Its boundary is the union of 2 vertical 
    lines
    (corresponding to $i = 0$ and $i = n$) and a finite number of segments, with
    slopes $q_i$, with $0 \leqslant q_1 < q_2 \ldots < q_r$.
    
    In the same way, the NP at $\infty$ of equation \eqref{e:edolinr} is the 
    convex hull of
    the set $\{ (u, v), (i, q) \in \mathcal{S}, 0 \leqslant u \leqslant i, v
    \leqslant q \}$ in $\mathbb{R}^+\times \mathbb{R}$.
\end{definition}

It is readily checked\cite{LR17} that the NP at $\infty$ of
equation \eqref{e:edolinr} is the symmetric with respect to the horizontal axis 
of the NP
at $0$ of the equation which is obtained from \eqref{e:edolinr} by the change 
of variable
$x \longrightarrow 1 / x$.

In the case of Euler's equation in homogeneous form \eqref{e:eulerlin}, the 
corresponding operator is 
\begin{equation*}
x\theta^2+\theta-1
\end{equation*}
and its NP
has 3 points $(0, 0)$, $(1, 0)$ and $(2, 1)$ so that we get the NP depicted in
Fig.~\ref{fig:NP_Euler} and the slopes are thus $0$ and $1$.

\begin{figure}
    \begin{center}
        \includegraphics{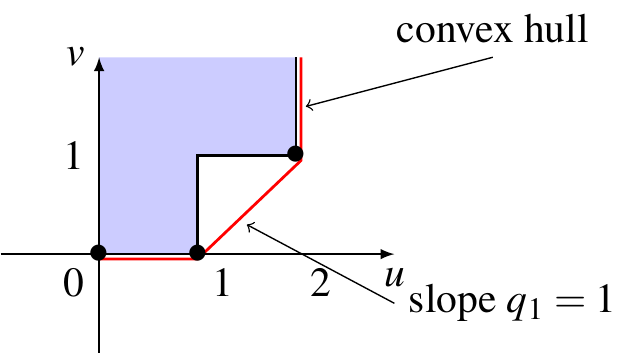}
    \end{center}
    \caption{\label{fig:NP_Euler}
        Newton Polygon at $0$ for Euler's equation \eqref{e:eulerlin}.
    }
\end{figure}

We shall notably use in the present article the so--called ``main theorem for
Gevrey asymptotics'', which can be stated in the following way:

\begin{theorem}
    Let $H \in \mathbb{C} (x) [d / d x]$. The equation $H (f) = 0$ will have a
    formal series solution $f$ if and only if its NP at $0$ has a slope which is
    horizontal. In that case, if there is no positive slope, then the formal 
    series solution is analytic; else, if we denote by $q_1 < \ldots < q_r$ the 
    positive
    slopes of the Newton polygon of the equation at the origin, the degrees of
    the determining factors appearing in a {basis of formal solutions} belong 
    to the set $\{
    q_1, \ldots, q_r \}$.
    
    Moreover, in any direction $d$, there is a sector $S$ at the origin,
    bisected by $d$, such that $H (f) = 0$ admits in $S$ a fundamental system of
    analytic solutions and the maximal growth rate of a solution approaching the
    origin is of type $|x|^{\beta} e^{\frac{\alpha}{|x|^{q_r}}}$ 
    ($\alpha \in \mathbb{R}^+$).\end{theorem}

Whenever the Borel transform $\widehat{f}$ of a formal series satifies a linear 
ODE in the variable $\zeta$, that is $\widehat{H}(\widehat{f})=0$ (as this will 
be the case in our examples), this theorem is useful to derive the growth at 
infinity of this function. If the NP at $\zeta=\infty$ of $\widehat{H}$ has 
slopes $0\geqslant q'_1 >\dots >q'_r$ then, in any nonsingular direction, the 
growth of $\widehat{f}$ is of type $|\zeta|^{\gamma}e^{\alpha|\zeta|^{-q'_r}}$ 
($\gamma\in \mathbb{R}$, $\alpha\in \mathbb{R}^+$).

The general algorithm that computes a basis of formal solutions is very clearly
exposed in Ref.~\onlinecite{Della} and it turns out that, for the Euler
equation, as well as for the equations $(E_k)$, the algorithm simplifies
drastically.
For these equations we can luckily apply the following procedure.

Suppose the Newton polygon of the differential operator
\[ H = H_n (x) \theta^n + \ldots + H_1 (x) \theta + H_0 (x)\]
has one horizontal slope from $(0,d_0)$ to $(1,d_0)$ and then one positive 
slope $q$ from $(1,d_0)$ to $(n,d_1)$ ($q=\frac{d_1-d_0}{n-1}$) then
\begin{itemize}
    \item[Step 1] The equation $H(f)$ has exactly one-dimensional family of
    solutions generated by some $y_0$ in $x^{\beta}\mathbb{C}[[x]]$ and the
    exponent $\beta$ is determined by the ``indicial equation''
    \[H_{1,d_0}\beta + H_{0,d_0}=0,\]
    where $H_{1,d_0}$ (resp. $H_{0,d_0}$) is the coefficient of degree $d_0$
    in $H_1$ (resp. $H_0$). More generally, the indicial equation
    $Q(\beta)=0$ where the polynomial $Q(\beta)$ corresponds to the
    coefficient of the monomial of lowest degree in $x$ in
    $x^{-\beta}H x^{\beta}$. Note that the $H x^{\beta}$ is the operator that 
    associates to any test function $f$ the function $H(x^{\beta}f)$ and not 
    the action of $H$ on the function $x^{\beta}$ that would be noted 
    $H(x^{\beta})$. 
    \item[Step 2] For $u\in \mathbb{C}$ one compute the operator
    \begin{align*}
    H_u
    & = e^{-\frac{u}{x^q}}H e^{\frac{u}{x^q}}\\
    & = H_n (x) (\theta-qux^{-q})^n + \ldots 
    + H_1 (x) (\theta-qux^{-q}) + H_0 (x)\\
    & = H_n (x,u) \theta^n + \ldots + H_1 (x,u) \theta + H_0 (x,u).
    \end{align*}
    The term of lowest degree in $x$ in $H_0 (x,u)$ is a polynomial $P(u)$
    of degree $n$, with 0 as a root.
    \item[Step 3] If this polynomial has exactly $n-1$ non zero distinct roots
    $u_1,\dots,u_{n-1}$, for each $u_i$, the newton polygon of $H_{u_i}$
    has the same shape as for $H$ and one can apply Step 1 to get a formal
    solution $y_i$ of $H_{u_i}$.
\end{itemize}
Once this procedure is finished, we obtain a basis of formal solutions:
\[ y_0(x),\ e^{\frac{u_1}{x^q}}y_1(x),\
\dots,\ e^{\frac{u_{n-1}}{x^q}}y_{n-1}(x).  \]
Note that the situation is less simple whenever the polynomial $P(u)$ has
multiple roots\cite{Della}.

\subsection{An exercise: the homogeneous Euler equation}\label{euler}

For $H=x\theta^2+\theta-1$ (see (\ref{e:eulerlin})), Step 1 in the algorithm 
ensures that we have a nontrivial solution $\widetilde{f_0}\in 
x^{\beta}\mathbb{C}[[x]]$ where the indicial equation is in this case 
$\beta-1=0$, thus $\beta=1$ and $\widetilde{f_0}\in x\mathbb{C}[[x]]$.
Step 2, with slope $q=1$, gives
\[
H_u = e^{-\frac{u}{x}}H e^{\frac{u}{x}}
=x\theta^2+(1-2u)\theta +[(u-1)+(u^2-u)x^{-1}]
\]
and $P(u)=u^2-u$ so that we can move to Step 3 with $u=1$ and 
$H_1=x\theta^2-\theta$. The algorithm ensures that there exists a formal 
solution to $H_1(f)=0$, here $y_1(x)=1$ for instance. We have a basis of 
solution $\widetilde{f_0},e^{\frac{1}{x}}$.

As we have already noticed, Newton's polygons are also informative on
the Borel transform of a solution and its growth at infinity. Toward a process
described in section \ref{sec:holoborel} (see also Ref.~\onlinecite{VDH07}), the
Borel transform  of $\widetilde{f_0}$ is itself  a solution of an equation 
$\widehat{H}(\widehat{f}_0)=0$ where \[
\widehat{H}=(1+\zeta)\theta +\zeta \qquad \text{with} \qquad \theta = \zeta 
\frac{d}{d \zeta} \]
and without any computation of the solutions, the Newton's polygon of this 
equation gives once again many informations:
\begin{enumerate}
    \item The NP at $\zeta=0$ has no positive slope: this equation has analytic 
    solutions at $\zeta=0$ (the indicial equation gives $\beta=0$) and the 
    Cauchy-Lipschitz theorem ensure that the solutions can be continuated along 
    any path that avoids $\zeta=-1$.
    \item The NP at $\zeta=\infty$ has no negative slope, that is to say that 
    its symmetric with respect to the horizontal axis, that corresponds to the 
    Newton's polygon at $0$ in the variable $\xi=1/\zeta$ has only a zero 
    slope: in the variable $\xi$, solutions near $\xi=0$ are in 
    $\xi^{\beta}\mathbb{C} \lbrace \xi \rbrace$ for some $\beta\in \mathbb{R}$ 
    and thus the solutions of $\widehat{H}(\widehat{f}_0)=0$ are $O(|\zeta 
    |^{-\beta})$ near $\zeta=\infty$. This ensures that the Laplace Transform 
    exists in any direction avoiding $\zeta=-1$.
\end{enumerate}

This interplay between the Borel Transform and Newton's polygon is very useful 
to understand, with few computations, the resurgence of solutions of ODEs. 
 
\section{Resurgent functions}\label{s:res}

\subsection{Algebras of resurgent functions and resurgence}

We recall now the main features of resurgent functions theory, \emph{first in an
    informal way}; special cases shall then be introduced next, with the precise
definitions for the restricted spaces of resurgent functions which are relevant
in the present article.

At its core, resurgence involves analytic functions of one complex variable
$\zeta$ which have isolated singular points and families of operators which
``measure'' these isolated singularities. There are spaces of resurgent
functions of various levels of complexity, the definitions of which shall
depend of the complexity of the problems they are involved in. In the present
article we shall only need to consider resurgent functions of a very
elementary type, introduced below.

We first consider the linear space $\mathcal{E}$ of germs $\varphi (\zeta)$ at
the origin of $\mathbb{C}$ which can be continued along any broken line
$\Gamma$ starting at $0$, possibly circumventing a finite number of isolated
singular points $\omega_i$ met on $\Gamma$. This definition presupposes that
the continuation of $\varphi$ close to such a point $\omega$ is defined for
paths going around it, following a small half--circle on the right or on the
left of the direction $d$ and, in general, this process will entail
multivaluedness.

With the convolution product $\varphi \ast \psi (\zeta) := \int_0^{\zeta}
\varphi (s) \psi (\zeta - s) d s$ (for $\zeta$ close to $0$ ), $\mathcal{E}$
is in fact an algebra (a highly non trivial fact \cite{E1,E3,sauzin}); 
multiplication by $-
\zeta$ is a derivation with respect to the convolution product and
$\mathcal{E}$ has a structure of differential algebra, with these operations.

We shall however very soon need to consider ``germs that are in fact also
singular at the origin'', namely functions defined in the lift of a pointed disk
at the origin of $\mathbb{C}$ on the Riemann surface of the logarithm
$\mathbb{C}_{\bullet}$.

By germ at the origin of $\mathbb{C}_{\bullet}$, we shall mean in this article 
a function
defined on the lift on $\mathbb{C}_{\bullet}$ of a pointed disk $D^{\ast} (0,
r)$ on $\mathbb{C}$. A regular germ $\varphi (\dot{\zeta})$
can be seen as the (class of the) singular germ $\check{\varphi} (\zeta) = 
\frac{\log
    (\zeta)}{2 i \pi} \varphi (\dot{\zeta})$

\begin{definition}
    Let $\mathcal{F}$ be the set of classes
    ${\varphi}^\nabla (\zeta)$ (mod $\mathbb{C} \{ \zeta \}$) of
    germs $\check{\varphi} (\zeta)$ on $\mathbb{C}_{\bullet}$ \ such that the 
    germ 
    $\hat{\varphi} (\zeta)$ has isolated singularities, where
    \[ \hat{\varphi} (\zeta) = \check{\varphi} (\zeta) - \check{\varphi} (e^{- 2
        \pi i} \zeta) \]
    $\check{\varphi} (\zeta)$ is called \textbf{a major} of $\varphi$ and
    $\hat{\varphi} (\zeta)$ is called \textbf{the minor} of $\varphi$.
    Elements of $\mathcal{F}$ will be called resurgent functions and will 
    simply be denoted
    by $\varphi $ when the abuse of language is innocuous.
\end{definition}

The set $\mathcal{F}$ carries a natural structure of vector space but in
fact it is an algebra, with a suitable definition of
convolution\cite{E3,sauzin}.

Elements $\varphi$ of $\mathcal{F}$ for which a major $\check{\varphi}$
satisfies $\zeta \check{\varphi} (\zeta) \longrightarrow 0$ when $\zeta
\longrightarrow 0$ uniformly in sectors of bounded opening and such that
$\hat{\varphi}$ is integrable at 0 are called \emph{integrable resurgent
    functions} and their set will be denoted by $\mathcal{F}_0$. These 
    resurgent functions are characterized by their minors and the convolution
alluded to above boils down to the convolution of minors:
we can identify such a function $\varphi$ with its minor $\hat{\varphi}$ and 
accordingly we shall have,
for any $\varphi$, $\psi$ in $\mathcal{F}$:
\[ {\varphi \ast \psi} (\zeta) := \int_0^{\zeta}
\hat{\varphi} (s) \hat{\psi} (\zeta - s) d s, \quad {(\zeta \sim 0 )}.
\]

In the applications considered by the present paper, it will always be
possible to work with integrable resurgent functions, possibly after some
simple transformation (essentially, for formal series, premultiplication of
$\varphi (z)$ by some suitable power $z^n$   ($n \in \mathbb{N}$), which
corresponds to performing a finite number of integrations in the Borel plane.

Generally, (see e.g. below or Refs.~\onlinecite{E3,sauzin}), resurgent 
functions which
appear as Borel transforms of divergent series which are solution to some 
complex
analytic dynamical system at a singular point will display exponential growth
in directions of the Borel plane which don't contain singularities. As such,
they will be amenable to the Borel--Laplace summation mechanism and we thus
have, in Ecalle's terminology three models for a resurgent function $\varphi$:

\begin{enumerate}
    \item space of formal series $\widetilde{\varphi (z)}$;
    \item space of classes of singular germs in the Borel plane
    ${\varphi^{\nabla} (\zeta)}$;
    \item space(s) of analytic functions $\varphi (z)$ in sector(s).
\end{enumerate}

We have given above the definition of Borel transform for formal series
without constant terms first, by:
\[ \mathcal{B} \left(\sum_{n=0}^\infty\  a_n z^{- n - 1} \right) =
\sum_{n=0}^\infty\  \frac{a_n}{\Gamma (n + 1)} \zeta^n \]
and then extended it to any formal series by defining $\mathcal{B} (1) =
\delta$.

As we have seen, the Major/Minor formalism very naturally incorporates the
convolution unit by expressing $\delta$ as the class determined by the major
$\frac{1}{2 \pi i \zeta}$, without making it necessary to consider spaces of
distributions -- and this is particularly valuable for questions of
convergence of sequences of resurgent functions.

Now, for the applications, it is necessary to enhance the definitions of Borel
and Laplace transforms to more general classes of formal objects: indeed,
even for the case of linear differential equations considered in the present 
work, we meet in the formal
solutions expansions involving ramified (Puiseux) formal series, together with
integer powers of logarithms.

A first extension consists thus in considering power functions, of any
exponent; we have to define the corresponding classes $\varphi^{\nabla}$ of
$\widetilde{\varphi} = z^u$ in such a way as to respect the properties of
morphisms of differential algebras, which are absolutely crucial. In the
following definition, it is necessary to make a distinction between integer
and non integer exponents:

\begin{definition}
    \begin{enumerate}
        \item For $\sigma \not\in \mathbb{Z}$, we define $\mathcal{B} (z^{- 
        \sigma})
        = \varphi^{\nabla}$ with, as pair (natural major, minor):
        \[ \left( \check{\varphi} = (1 - e^{- 2 \pi i \sigma}) 
        \frac{\zeta^{\sigma
                - 1}}{\Gamma (\sigma)}, \hat{\varphi} = \frac{\zeta^{\sigma -
                1}}{\Gamma (\sigma)} \right). \]
        \item For $n \in \mathbb{N}$, we define $\mathcal{B} (z^n) =
        \varphi^{\nabla}$ with, as pair (natural major, minor):
        \[ \left( \check{\varphi} = \frac{(- 1)^n}{2 \pi i} \zeta^{- n - 1} 
        \Gamma
        (n + 1), \hat{\varphi} = 0 \right). \]
    \end{enumerate}
\end{definition}

For any $\sigma \in \mathbb{C}$, with \ $\Re\sigma > 0$, the Borel
transform of $z^{- \sigma}$ is characterized by its minor (it is an integrable
singularity) and we can safely denote in this case $\mathcal{B} (z^{- \sigma})
= \frac{\zeta^{\sigma - 1}}{\Gamma (\sigma)}$, using the same abuse of
notations as for formal series with integer exponents. This can be extended
formally to:
\[ \mathcal{B} \left( \sum_{\sigma \in \mathbb{Q}^{>}} c_{_{\sigma}} z^{-
    \sigma} \right) = \sum_{\sigma \in \mathbb{Q}^{>}} c_{\sigma} 
\frac{\zeta^{\sigma - 1}}{\Gamma (\sigma)}. \]
In expansions using Puiseux series, if we have geometrical growth estimates
($\left|c_{\sigma}/\Gamma (\sigma)\right| \leqslant A B^{\sigma} ; A, B >
0$, for exponents $\sigma$ which are multiples of a given rational number, 
say), then the right hand side defines a germ (on $\mathbb{C}_{\infty}$) and
the Borel transform indeed extends as a morphism of differential algebras from
the space of Puiseux formal series in $z^{- 1}$ with the ordinary product and
derivation $\partial_z = \frac{d}{d z}$ to the algebra of local integrable
resurgent, with convolution of minors and as derivation the multiplication (of
minors) by $- \zeta$.

More generally, for any $\sigma \in \mathbb{C}$, with \ $\Re \sigma > 0$
and $r \in \mathbb{N}$, the Borel transform of $z^{- \sigma} (\log z)^r$ is
the integrable resurgent function characterized by its minor, which is:
\[ \mathcal{B} (z^{- \sigma} (\log z)^r) = \zeta^{\sigma - 1} \sum_{i = 0}^r
\binom{r}{i} \left( \frac{1}{\Gamma (\zeta)} \right)^{(r)} (\log
\zeta)^{r - i}. \]

\subsection{Alien operators}

Let us consider a resurgent function $\varphi^{\nabla}$, as in the previous 
subsection,
which is characterized by its minor $\hat{\varphi}$; we will moreover suppose
first that $\hat{\varphi}$ is regular at the origin of $\mathbb{C}$ (this will
be the case for the applications to the equations $E_k$ below), and denote by
$\Omega$ its set of singularities in $\mathbb{C}$. We shall also suppose that
is a fixed discrete set of $\mathbb{C}$, which will be enough to introduce
all the necessary concepts; in the applications below, $\Omega$ will in fact
be finite.

For any $\omega$ in $\mathbb{C}$ and any path $\gamma$ from the origin to
$\omega$ there is an operator $\Delta_{\omega}^{\gamma}$ that measures the
singularity at $\omega$ of the analytic continuation along $\gamma$ of any
element of $\mathcal{E}$:
\[ \Delta_{\omega}^{\gamma} \varphi (\zeta) := {\varphi_{\omega}
    (\omega + \zeta)} \mod \mathbb{C} \{ \zeta \}. \]
In this formula,
$\varphi_{\omega}$ is the analytic continuation of $\varphi$ along $\gamma$,
which is defined on the lift on $\mathcal{S}$ of a pointed disk $D (\omega,
r)$.

It is important to observe that this definition involves 2 operations:
``extraction of singularity'' at $\omega$ and translation to the origin. The
operator $\Delta_{\omega}^{\gamma}$ associates to any given regular germ at 0
(element of $\mathcal{E}$) a class of singular germs. The
$\Delta_{\omega}^{\gamma}$ are linear and we will obtain endomorphisms if we
enhance the following definition to the vector space $\mathcal{F}$: if
$\varphi^{\nabla}$ is an element of $\mathcal{F}$,
$\Delta_{\omega}^{\gamma} \varphi$ is given by the same formula above, in
which $\varphi_{\omega}$ is by definition now the analytic continuation of
\emph{the minor} of $\varphi$ along $\gamma$.

Averages of some of these $\Delta_{\omega}^{\gamma}$, with coefficients which 
satisfy relevant symmetry properties\cite{E1, E3, sauzin} will next yield
the alien derivations (there are several families of them; we shall only need 
the standard one). When $\omega$ is
the first singularity met in the continuation in a given direction $d$
(``polarization''), we define: \ $\Delta_{\omega} \varphi (\zeta) : =
\Delta_{\omega}^{\gamma} \varphi (\zeta)$, where $\gamma$ is the segment on
$d$ from the origin to $\omega$.

When there are several singularities $\omega_1, \ldots, \omega_{r - 1}$ on $d$
before reaching $\omega_r = \omega$, we define
\[ \Delta^+_{\omega} \varphi (\zeta) : = \Delta_{\omega}^{\gamma} \varphi
(\zeta), \]
where $\gamma$ is a path from the origin to $\omega$ closely following $d$
from the right.

Finally, there are families of numbers $A^{(\omega_1, \ldots, \omega_r)}$
such that the operators defined by
\[ \Delta_{\omega} \varphi (\zeta) : = \sum_{\| (\omega_1, \ldots, \omega_r)
    \| = \omega} A^{(\omega_1, \ldots, \omega_r)} \Delta^+_{\omega_r} \ldots
\Delta^+_{\omega_1} \quad \text{where}\  \| (\omega_1, \ldots, \omega_r) \|
= \omega_1 + \ldots + \omega_r \]
are indeed derivations of the algebra $\mathcal{F}$. The symmetry property for
the $A^{(\omega_1, \ldots, \omega_r)}$ which will ensure that the associated
$\Delta_{\omega}$ will satisfy Leibniz rule is called alternelity but we won't
need to go into these considerations in the present text because we shall only
consider the simplest of these families which is given by
$A^{(\omega_1, \ldots, \omega_r)} =
\frac{(- 1)^{r - 1}}{r}$, which yields the so--called standard alien
derivations.

When the resurgent functions have corresponding expressions $\widetilde{\varphi}
(z)$ ``in the formal model'' (it is e.g. the case for majors of the type
$\zeta^{\sigma} (\log \zeta)^k$, as we have seen above, then \ pullbacks by
$\mathcal{B}^{- 1}$ then gives alien operators acting on the formal objects
$\widetilde{\varphi} (z)$, for which we keep the same notation
$\Delta_{\omega}$.

Finally, we introduce the exponential--carrying alien operators
$\pmb{\Delta}_{\omega}$ (denoted in bold, or also as pointed alien
derivations in earlier texts) by the following expression in the formal model
\[ \pmb{\Delta}_{\omega}  \widetilde{\varphi} (z) := e^{- \omega z}
\Delta_{\omega}  \widetilde{\varphi}  (z). \]

These $\pmb{\Delta}_{\omega}$ act on formal expressions involving not
only resurgent functions but also exponentials $e^{- \omega z}$; there is a
thorough theory of these so--called \emph{transseries} but in the present
work we shall only cope with finite sums of terms
$e^{- \omega z}\widetilde{\psi}(z)$, typically involving a basis of the finite
dimensional vector spaces of the solutions of the linear ODEs we are dealing
with. The
$\pmb{\Delta}_{\omega}$ commute with $\partial = \frac{\partial}{\partial
    z}$, which makes them particularly convenient in formal calculations, as we
shall see in practice below.

\subsection{Resurgence for holonomic functions with one critical time}
\label{sec:holoborel}

\begin{definition}
    A function of one complex variable $x$ or a (Puiseux) formal series $f (x)$
    is called holonomic if it is solution to a linear ODE
    \[ H (f) = 0 \quad \text{where} \quad H \in \mathbb{C} [x] \left[
    \frac{d}{d x} \right]. \]
\end{definition}

An equivalent definition would be to require the existence of $H \in
\mathbb{C} (x) \frac{d}{d x}$ (which can then chosen to be monic) such that
$H (f) = 0$ and a convenient characterization of a holonomic function is that
$\mathbb{C} \left[ \frac{d}{d x} \right] (f)$ be of finite type over
$\mathbb{C} (x)$.

A function which is a solution of an equation ``of affine type'', say $H (f) =
g$ (with $H \in \mathbb{C} (x) \left[ \frac{d}{d x} \right]$ and $g \in
\mathbb{C} (x)$ ) is holonomic, as seen by dividing the equation by $g$ and
applying $\frac{d}{d x}$. Thus, solutions of Euler's equation above are
holonomic.

We recall the very well known following stability properties, for holonomic
functions:

\begin{proposition}
    The sum and the product of two holonomic functions are holonomic. The
    postcomposition of a holonomic function by an algebraic function is
    holonomic
\end{proposition}

These properties are in fact true for holonomic functions of several variables
and rely on the characterization of holonomy mentioned above (for algorithmic
aspects see e.g. Ref.~\onlinecite{VDH07}).
In particular, if $f (x)$ is holonomic, then, for any rational number $r$,
$g(x) := f (x^r)$ is holonomic and this simple
result will be crucial for questions of resurgence in our context.

\begin{proposition}
    The Borel transform of a holonomic function $f \in x\mathbb{C}
    [[x^{\mathbb{Q}^+}]]$ with vanishing constant coefficient is holonomic. As
    a consequence it is resurgent, with a finite number of singularities, by
    applying to the Borel plane the global theory of linear ODEs.
\end{proposition}

\begin{proof}
    If $H (f) (x) = a_n (x) f^{(n)} (x) + \ldots + a_1 f' (x) + a_0 (x) f (x) =
    0$, we can rewrite this equation using the derivation $D$ defined by $D (f)
    (x) = x^2 f' (x)$, possibly after having multiplied it by a suitable integer
    power of $x$ to get an equation $K (f) = 0$, and then use the fact that the
    Borel transform of $D (f) (x)$ is $\zeta \hat{f} (\zeta)$ and that for any
    series $f \in x\mathbb{C} [[x^{\mathbb{Q}^+}]]$ of valuation $> 1$,
    $\mathcal{B} (x^{- 1} f) (\zeta)$=$\hat{f}' (\zeta)$.
    
    As a consequence, it is resurgent with a finite number of singularities, by
    applying to the Borel plane the global theory of linear ODEs (for these
    matters, see also Ref.~\onlinecite{ANDRE00}).
\end{proof}

Thus, from equation $(E_2)$ expressed in the variable $x$, we obtain
successively
\begin{equation}\label{eq:fhat}
\begin{aligned}
16 x^2 f'' (x) + 32 x^3 f' (x) + 3 f (x) + f' (x) & =  0\\
16 D^2 (f) (x) + 3 x^2 f (x) + D (f) (x) & =  0\\
x^{- 2} (16 D^2 (f) (x) + D (f) (x)) + 3 f (x) & =  0\\
((16 \zeta^2 + \zeta) \hat{f}  (\zeta))'' + 3 \hat{f}  (\zeta) & =  0\\
(16 \zeta^2 + \zeta) \hat{f}'' (\zeta) + 2 (32 \zeta + 1) \hat{f}'
(\zeta) + 35 \hat{f}  (\zeta) & =  0
\end{aligned}
\end{equation}
The last equation is singular at $0$ and $- \frac{1}{16}$ and this entails
that $\hat{f}$ can only have a singularity at $-
\frac{1}{16}$.

\begin{theorem}
    Let $f$ be a formal series which is solution to $H (f) = 0$, where $H \in
    \mathbb{C} [x] (d / d x)$ has a Newton polygon at $0$ with a single non
    zero slope, equal to $q$.
    
    Then $f$ is resurgent with respect to the critical time $z = 1 / x^q$, with
    a finite number of singularities in the Borel plane and exponential growth
    of order 1 in any non--singular direction
\end{theorem}

\begin{proof}
    By premultiplying $f$ by a suitable power of $x$ with a positive exponent,
    we can suppose that $\operatorname{val} (f) > q$; this operation doesn't
    change the holonomic character, nor the slopes of the Newton polygon.
    
    Let $g (x) := f \left( x^{\frac{1}{q}} \right)$. Then $g$ is holonomic
    and the Newton polygon at $0$ of the differential equation $U (g) = 0$, with
    $U \in \mathbb{C} [x] (d / d x)$ \ obtained from $H (f) = 0$ by the action
    of ramification $x \longrightarrow x^{\frac{1}{q}} $ has a single slope,
    equal to one.
    
    The series $g$ is Gevrey--1 and its Borel transform is holonomic, by the
    previous proposition ($\operatorname{val} (g) > 1$). Moreover, the NP of the
    Borel
    transform of $U$ has a single slope equal to one at $\infty$
    (see Ref.~\onlinecite{LR17}),
    which entails that $\hat{g} (\zeta)$ has at most exponential growth of order
    $1$ at $\infty$ and thus $\hat{g}$ is Laplace--summable in every non
    singular direction $d$ ($d \cap S = \varnothing$, where $S$ designates the
    finite set of singularities of $\hat{g}$).
    
    By reverting to the original variable, we get sectorial solutions on
    sectors of opening $\pi / q$ for $f (x)$.
\end{proof}

All the formal series considered in the present text will satisfy the
hypotheses of this theorem and we shall see how the singularities of their
Borel transforms can be analyzed by the action of alien operators introduced
in the previous subsection, eventually providing a control of Stokes
phenomenon for the sectorial sums of the series.

The algebra of resurgent functions which we need for the present article is
simply the algebra $\mathcal{R}_h$ of holonomic functions $f$: for any $q \in
\mathbb{Q}^{>}$, the function $g (x) := f \left( x^{\frac{1}{q}}
\right)$ has a Borel transform which is characterized by its minor, which has
isolated singularities and $g$ is thus resurgent and amenable to alien
calculus.
 
\section{Resurgent study of the family $(E_k)$}\label{s:resEk}

\subsection{The first equation}

All the equations in the family $(E_k)_{k \geqslant 2}$ will be tractable by
the same techniques, yet it is worthwhile to start by an exhaustive study of
$E_2$ because the resurgence properties will be accessible in a
straightforward way.

For $k = 2$, we have thus obtained above the following linear ODE, with
$\partial_{\lambda} = \frac{d}{d \lambda}$:
\[ [(4 \lambda \partial_{\lambda} + 1) (4 \lambda \partial_{\lambda} + 3) +
\partial_{\lambda}] f (\lambda) = 0 \qquad (E_2) \]
For any given $a_0$ this equation has a unique solution in $\mathbb{C}
[[\lambda]]$, namely
$\widetilde{f} (\lambda) = \sum_{n=0}^\infty a_n \lambda^n$, with
\[ a_{n + 1} = - \frac{(4 n + 1) (4 n + 3)}{n + 1} a_n \]
and $\widetilde{Z_0} (\lambda)$ is the solution of $(E_2)$ in $\mathbb{C}
[[\lambda]]$ with $a_0 = 1$.

This recurrence relation immediately entails (for any $a_0 \neq 0$) the
\emph{exact} Gevrey--1 rate of growth for the coefficients $a_n$, with:
\[ | a_n | \sim 16^n n! | a_0 |. \]
From this estimate, we can already deduce that
$\sum_{n=0}^\infty a_n \lambda^n$ is indeed
divergent; it has $z = \frac{1}{\lambda}$ as critical time and the distance to
the origin of the closest singularity in the Borel plane of
$\mathcal{B}\left(\widetilde{Z_0}\right)$ is $\frac{1}{16}$ but
considerations on the Newton's polygon and the \emph{Borel transform of
    equation} $(E_2)$ will easily
yield a more precise results.

If $x=\lambda$ and $\theta=x\frac{d}{dx}$ the operator associated to $(E_2)$ is
\[ H_2=16\theta^2 +(16+x^{-1})\theta +3 \]
that corresponds to the Newton's polygon drawn in Fig.~\ref{fig:NP_H2}.

\begin{figure}
    \begin{center}
        \includegraphics{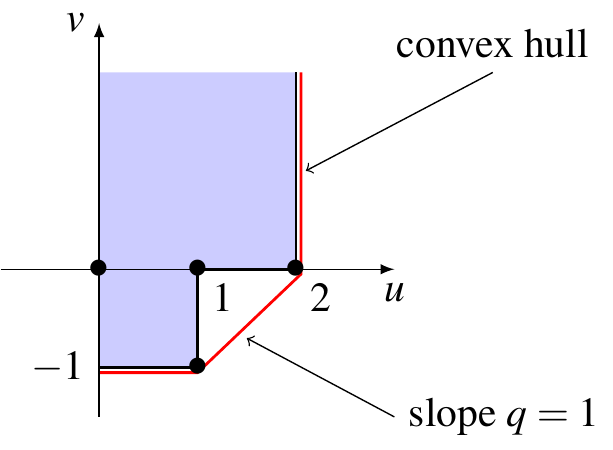}
    \end{center}
    \caption{\label{fig:NP_H2}
        NP corresponding to $H_2 = 16\theta^2 + (16 + x^{-1})\theta + 3$.
    }
\end{figure}

The algorithm in section \ref{sec:NP} applies here :
\begin{itemize}
    \item[Step 1] The indicial equation gives $\beta=0$: $(E_2)$ as a 
    nontrivial formal solution in $\mathbb{C}[[x]]$, for example 
    $\widetilde{Z_0}(x)$.
    \item[Step 2] The unique positive slope is $q=1$ and
    \begin{align*}
    H_{2,u}&= e^{-\frac{u}{x}}H_2 e^{\frac{u}{x}}\\
    & = 16(\theta-ux^{-1})^2 + (16+x^{-1}) (\theta-ux^{-1}) + 3 \\
    &= 16 \theta^2 +(16+(1-32u)x^{-1})\theta + 3+ u(16u-1)x^{-2},
    \end{align*}
    thus $P(u)=u(16u-1)$ with $u_1 = \frac{1}{16}$ as a root.
    \item[Step 3] We can complete the basis of formal solutions by 
    $e^{\frac{1}{16x}}\widetilde{Z}_1$ where $\widetilde{Z}_1$ cancels out the 
    operator:
    \begin{equation*}
    H_{2,\frac{1}{16}}=16 \theta^2 +(16-x^{-1})\theta + 3.
    \end{equation*}
\end{itemize}
For this latter equation, we can also find an explicit series as a solution: 
$\widetilde{Z}_1(\lambda) = \sum_{n=0}^\infty b_n \lambda^n$ with the following 
recurrence relation:
\[ b_{n + 1} = \frac{(4 n + 1) (4 n + 3)}{n + 1} b_n \]
and of course the same growth estimates for the $b_n$ as for the $a_n$.


We have seen above (see equation (\ref{eq:fhat})) that the Borel transform of 
$(E_2)$ can be written in the
following form:
\[ (16 \zeta^2 + \zeta) \hat{f}'' (\zeta) + 2 (32 \zeta + 1) \hat{f}' (\zeta) +
35 \hat{f}  (\zeta)  =  0. \]
By the general properties of the Borel transform, we have that $\widehat{Z_0}
(\zeta)$ cancels out the corresponding operator (up to a left multiplication by 
$\zeta$)
\[
\widehat{H}_2=(16\zeta+1)\theta^2 +(48\zeta+1)\theta +35 \zeta 
\] that has no singular point at $\zeta=0$: its solutions are analytically
continuable in $\mathbb{C}_{\zeta}$ with as only singularity the point
$\omega = - \frac{1}{16}$. 

As in section \ref{euler}, the Newton polygon of $\widehat{H}_2$ at $\infty$ 
has only a zero slope which entails that $\widehat{Z_0} (\zeta)$ has at most a 
polynomial growth at $\infty$, along any direction but the real
negative one.
As a consequence, $\widetilde{Z_0}$ is 1--summable in every direction except
$\mathbb{R}_{<}$ and there is only one alien derivation that can act non
trivially on it, namely $\Delta_{_{- \frac{1}{16}}}$.

For the solution corresponding to $H_{2,\frac{1}{16}}$, one can check that this 
is the same NP as for $H_2$ and that the Borel transform corresponds to the 
equation \[
(16 \zeta^2 - \zeta) \hat{f}'' (\zeta) + 2 (32 \zeta - 1) \hat{f}' (\zeta) +
35 \hat{f}  (\zeta)  =  0
\]
and we obtain that $\widehat{Z}_1$ is resurgent with a single singularity at
$\omega = \frac{1}{16}$ and exponential growth at infinity.

The general solution $\Phi (\lambda)$ of $(E_2)$ is thus:
\[ \Phi (\lambda) = \sigma_0  \widetilde{Z_0} (\lambda) + \sigma_1
e^{\frac{1}{16 \lambda}}  \widetilde{Z_1} (\lambda), \]
where $\sigma_i \in \mathbb{C}$.
We already know that the only alien derivations $\Delta_{\omega}$ which act
non trivially on $\widetilde{Z_0} (\lambda)$ and $\widetilde{Z_1} (\lambda)$
are respectively $\Delta_u$ and $\Delta_{- u}$ , with $u = - \frac{1}{16}$ but
a key remark is that we can also \emph{deduce that from the shape of the
    general solution above}.

Indeed, for any $\omega \in \mathbb{C}$, we have:
\begin{align*}
e^{- \frac{\omega}{\lambda}} \Delta_{\omega} \Phi (\lambda)
&= \sigma_0 e^{-
    \frac{\omega}{\lambda}} \Delta_{\omega} \widetilde{Z_0} (\lambda) +
\sigma_1 e^{- \frac{\omega}{\lambda}} \Delta_{\omega} e^{\frac{1}{16
        \lambda}}  \widetilde{Z_1} (\lambda) \\
&= \sigma_0 e^{-
    \frac{\omega}{\lambda}} \Delta_{\omega} \widetilde{Z_0} (\lambda) +
\sigma_1 e^{\frac{\frac{1}{16} - \omega}{\lambda}} \Delta_{\omega} 
\widetilde{Z_1} (\lambda)
\end{align*}
but $e^{- \frac{\omega}{\lambda}} \Delta_{\omega} \Phi (\lambda)$ is also a
solution of $(E_2)$, because $e^{- \frac{\omega}{\lambda}} \Delta_{\omega}$
commutes with $\partial_{\lambda}$ and vanishes on the polynomial coefficients
of the equation.

Thus, there exist 2 complex constants $c$ and $d$ such that:
\[ e^{- \frac{\omega}{\lambda}} \Delta_{\omega} \Phi (\lambda) = c
\widetilde{Z_0} (\lambda) + d e^{\frac{1}{16 \lambda}}  \widetilde{Z_1}
(\lambda) \]
and this \emph{entails} that the only alien derivations which can act non
trivially correspond to the 2 following indices:

$\omega = - u$, with \ $\Delta_{- u} \widetilde{Z_0} (\lambda) = A_{- u} 
\widetilde{Z_1} (\lambda)$ $(A_{- u} \in \mathbb{C})$ and $\Delta_{-
    u} \widetilde{Z_1} (\lambda) = 0$

$\omega = u$, with $\Delta_u \widetilde{Z_1} (\lambda) = A_u  \widetilde{Z_1}
(\lambda)$  $(A_u \in \mathbb{C})$ and $\Delta_u \widetilde{Z_1} (\lambda)
= 0$

This is an illustration, \emph{in an elementary situation}, of the power of
alien calculus: once we know that the series we are dealing with are
resurgent, we can easily get highly non trivial relations just by taking into
account formal rules and relations of homogeneity when applying alien
derivations to them.

The previous 2 resurgence relations can be expressed in the following more
compact form:
\begin{equation}\label{eq:be2} \pmb{\Delta}_{\omega} \Phi = \left( A_0 \sigma_0 
\frac{\partial}{\partial \sigma_0} + A_1 \sigma_1  \frac{\partial}{\partial
    \sigma_1} \right) \Phi 
\end{equation}
We see here that alien derivations act on the formal integral (the general
solution of $E_2$) as ordinary differential operators.

We thus have a bridge between alien and ordinary calculus and \eqref{eq:be2}
is an elementary instance of the \emph{bridge equation}, which is a very
general fact: for all the known cases of analysis of irregular singularities
of functional (differential, difference, etc) with the apparatus of resurgence
and alien calculus, some form of the bridge equation, depending upon the class
of equation under study, can be explicited and this fact has numerous
important consequences.

\begin{remark}
    The partition function for $k = 2$ is of course to be found in many articles
    and textbooks, for various illustrative purposes. \emph{Notably}, a very
    thorough resurgent study of it is thus done in section 2 of 
    Ref.~\onlinecite{ANICETO2019} -- with a
    slightly different normalization; as shown in this reference, $E_2$ belongs
    to the hypergeometric family (it is a modified Bessel eq.) and so does its
    Borel transform, which gives access to exact values of the coefficients
    $A_{\omega}$.
    
    It is however more important in the present work, rather than relying on
    properties of the hypergeometric family to stress the efficiency of the
    general formal arguments used above, which will work in a similar way for
    higher dimensional equations and also in non linear situations, where as a
    rule no exact values will exist for the coefficients of the bridge equation.
\end{remark}

\subsection{Formal results}

We start by the Newton polygon at $0$ of equation $(E_k)$ that corresponds to 
the operator
\begin{equation*}
H_k=\left(\prod_{j = 0}^{k - 1} (2 k \theta + 2 j + 1)
\right) + \lambda^{-1} \theta
\end{equation*}
which is given in Fig.~\ref{fig:NP_Hk_lambda}.
\begin{figure}
    \begin{center}
        \includegraphics{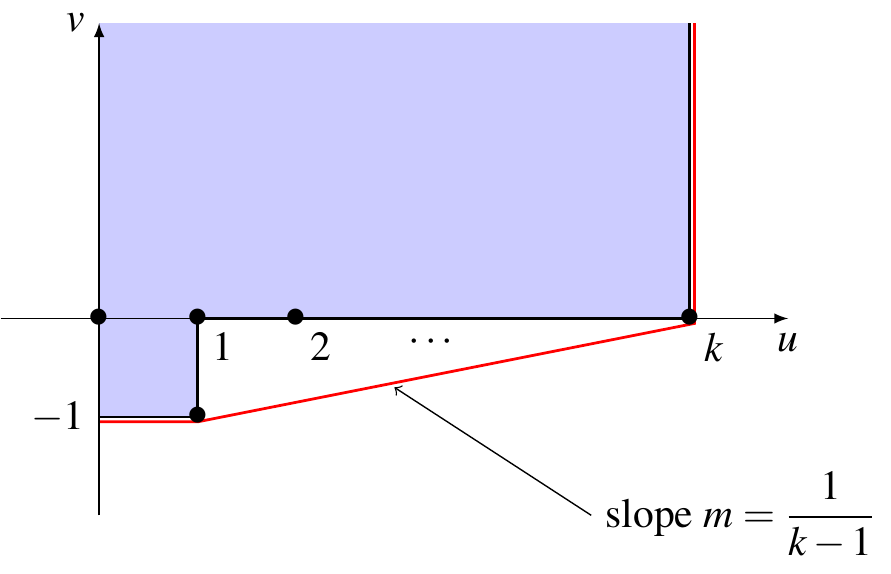}
    \end{center}
    \caption{\label{fig:NP_Hk_lambda}
        The NP for $H_k$ with $\lambda$ as variable
    }
\end{figure}
It has one horizontal slope, of length one, and only one non--zero slope, of
value $m = \frac{1}{k - 1}$

Accordingly, $(E_k)$ has one formal series solution $\widetilde{Z_0}
(\lambda)$ with a constant coefficient equal to $1$ (the indicial equation gives
$\beta=0$, which is Gevrey--1 \emph{with respect to the variable}
$z_k = \frac{1}{\lambda^m}$, as a consequence of the
propositions above).
This unique slope suggests to study the equation in the variable $x=\lambda^m$
($\lambda=x^{k-1}$) and since
\[
\theta_{\lambda}=\lambda\frac{\partial}{\partial 
\lambda}=mx\frac{\partial}{\partial x}=m\theta_x
\]
we will focus on the operator in the variable $x$:
\begin{equation*}
H_k = \left(\prod_{j = 0}^{k - 1} (2 km\theta_x + 2 j + 1) \right)
+mx^{1-k}\theta_x
\end{equation*}
whose Newton's polygon is given in Fig.~\ref{fig:NP_Hk_x} and has a single
non-zero slope $q=1$.

\begin{figure}
    \begin{center}
        \includegraphics{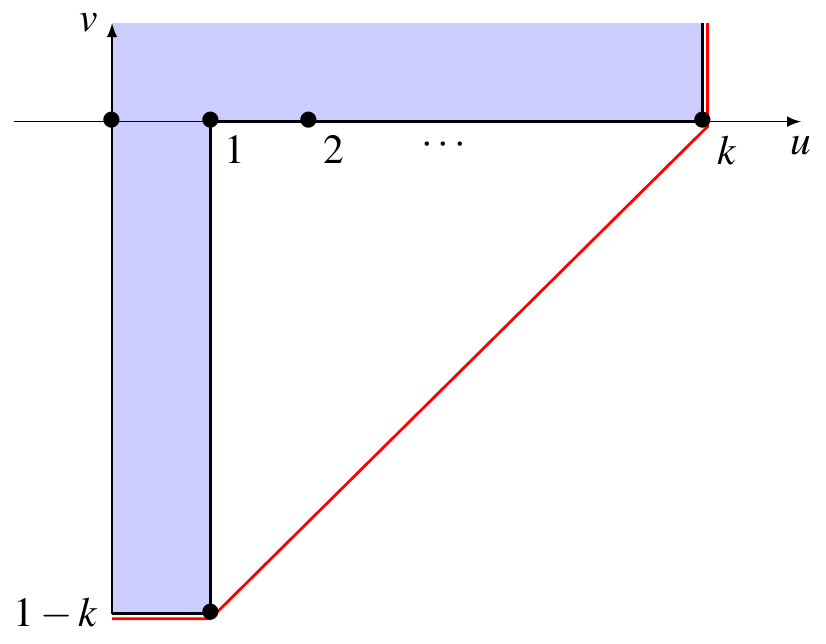}
    \end{center}
    \caption{\label{fig:NP_Hk_x}
        The NP for $H_k$ after choosing $x = \lambda^{\frac{1}{k-1}}$ as the new
        variable, notice its slope $q=1$.
    }
\end{figure}

In order to compute the operator $H_{k,u} = e^{-u/x}H_k e^{u/x}$ let us
introduce some notations and combinatorial coefficients. We first note, for
$k\geqslant 1$,
\[
\prod_{j = 0}^{k - 1} (X + 2 j + 1)=\sum_{i=0}^k a_{k,i} X^i
\]
and one can check that for any $k$, $a_{k,k}=1$ and $a_{k,k-1}=k^2$.
Second, let us define the Exponential (or Touchard) polynomials recursively by
$T_0(X)=1$ and, for $n\geqslant 1$,
\begin{equation*}
T_n(X)=\left(X+X\frac{d}{dX}\right)T_{n-1}(X)
\end{equation*}
the polynomial $T_n$ is of degree $n$ and, if $n\geqslant 1$ its coefficients
are the stirling numbers of second kind (see for instance 
Ref.~\onlinecite{Boyadzhiev}):
$$
T_{n}(X)=\sum_{k=1}^n s_{n,k}X^k
$$ for which we can notice that $s_{n,n}=1$ and $s_{n,n-1}=\frac{n(n-1)}{2}$. 
These polynomials will appear naturally in the sequel since:
\begin{equation*}
\forall n\geqslant 0,\quad e^{-u/x}\lbrack 
\theta_x^n.e^{u/x}\rbrack=(-1)^nT_n\left(\frac{u}{x}\right).
\end{equation*}

We perform now the change of unknown function $f (x) = e^{u/x} g(x)$, and use 
the previous coefficients:
\begin{align*}
e^{-\frac{u}{x}}H_k (e^{\frac{u}{x}}g(x))
&= \sum_{i=0}^k a_{k,i}e^{-\frac{u}{x}}(2km)^i\theta_x^i(e^{\frac{u}{x}}g(x))
+ e^{-\frac{u}{x}}mx^{1-k}\theta_x(e^{\frac{u}{x}}g(x)) \\
&= \sum_{i=0}^k a_{k,i}e^{-\frac{u}{x}}(2km)^i\sum_{j=0}^i\binom{i}{j}
(\theta_x^{i-j}(e^{\frac{u}{x}}))(\theta_x^{j}g(x))\\
& \quad +mx^{1-k}\theta_xg(x)-mux^{-k}g(x)\\
&= \sum_{i=0}^k a_{k,i}(2km)^i\sum_{j=0}^i\binom{i}{j} 
(-1)^{i-j}T_{i-j}\left(\frac{u}{x}\right)\theta_x^{j}g(x)\\
& \quad+mx^{1-k}\theta_xg(x)-mux^{-k}g(x)\\
&= \sum_{j=0}^k\sum_{i=j}^k a_{k,i}(2km)^i \binom{i}{j} (-1)^{i-j} 
T_{i-j}\left(\frac{u}{x}\right)\theta_x^{j}g(x)\\
&\quad+mx^{1-k}\theta_xg(x)-mux^{-k}g(x) \\
H_{k,u}g(x)&= \sum_{j=0}^k\left(\sum_{i=0}^{k-j} a_{k,j+i}(2km)^{i+j}
\binom{i+j}{j} (-1)^{i}T_{i}\left(\frac{u}{x}\right)\right)\theta_x^{j}g(x))\\
& \quad +mx^{1-k}\theta_xg(x)-mux^{-k}g(x)
\end{align*}

For a ``generic'' value of $u$, the operator $H_{k,u}$ has the Newton polygon
shown in Fig.~\ref{fig:NP_Hku}
\begin{figure}
    \begin{center}
        \includegraphics{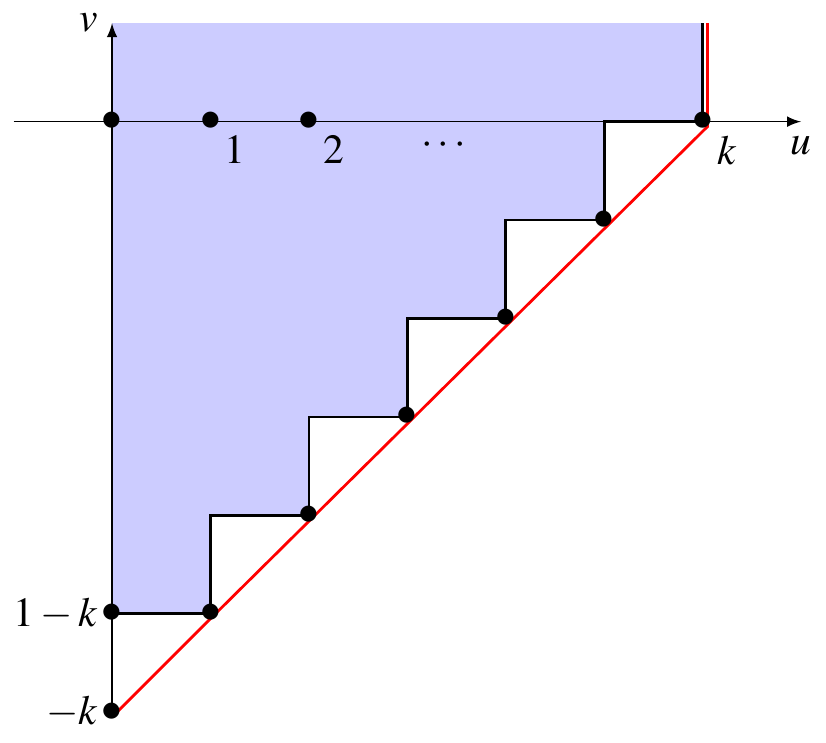}
    \end{center}
    \caption{\label{fig:NP_Hku}
        NP at $0$ of $H_{k,u} = e^{-\frac{u}{x}}H_ke^{\frac{u}{x}}$ for a
        ``generic'' $u\in\mathbb{C}$.
    }
\end{figure}
and we can specify three coefficients in the operator:
\begin{itemize}
    \item The coefficient of $x^{-k} \theta_x^0$ is 
    \[ P_k(u)=a_{k,k}(-1)^k (2km)^ks_{k,k} u^k -mu =-mu +(-2kmu)^k, \]
    that has only simple roots and for each of its roots, the point $(0,-k)$ 
    disappears.
    \item The coefficient of $x^{1-k}\theta_x$ is
    \begin{align*}
    Q_k(u)
    &= m+a_{k,k}(2km)^k \binom{k}{1} (-1)^{k-1} s_{k-1,k-1} u^{k-1} \\
    & =m+2k^2m(-2kmu)^{k-1}
    \end{align*}
    so that, whenever $P_k(u)$ vanishes, $Q_k(u)$ does not: for each root of
    $P_k(u)$, the operator $H_{k,u}$ has the same NP as $H_k$ so there exists a
    Gevrey--1 series (up to a factor $x^{\beta}$) that cancels out this 
    operator.
    \item The coefficient of $x^{1-k}\theta_x^0$ is
    \begin{multline*}
    a_{k,k-1}(-2km)^{k-1}s_{k-1,k-1}u^{k-1} +a_{k,k}(-2km)^k s_{k,k-1}u^{k-1} \\
    =(-2kmu)^{k-1}\left(k^2-2km\frac{k(k-1)}{2}\right)=0,
    \end{multline*}
    since $m=\frac{1}{k-1}$. This means that, for a root of $P_k(u)$ the 
    indicial equation (that corresponds to the coefficient of lowest degree in 
    $H_{k,u}.x^\beta$) is simply $\beta Q_k(u)=0$ thus $\beta=0$.
\end{itemize}

We can now state the following:

\begin{proposition}
    The equation $(E_k)$, in the variable $x=\lambda^m$ has as general solution:
    \[ 
    \Phi (x) = \sigma_0  \widetilde{\Phi_0} (x)
    + \sigma_1
    e^{\frac{u_1}{x}}  \widetilde{\Phi_1} (x) + \ldots + \sigma_{k - 1}
    e^{\frac{u_{k - 1}}{x} } \widetilde{\Phi_k} (x),\quad
    \sigma_0, \ldots \sigma_{k - 1} \in \mathbb{C} \]
    where:
    \begin{enumerate}
        \item $m = \frac{1}{k - 1}$.
        
        \item $u_0 = 0, u_1, \ldots, u_{k - 1}$ are the $k$ roots of the 
        polynomial
        $$P (u) = - m u [(- 1)^{k - 1} (2 k)^k (m u)^{k - 1} + 1].$$
        
        \item The formal series $\widetilde{\Phi_i} (x)$ are Gevrey of order 
        $1$.
        \item The formal series $\widetilde{Z_0} (\lambda)$ is proportionnal to
        $\widetilde{\Phi_0}(\lambda^m)$ and thus Gevrey of order
        $s =\frac{1}{m} = k - 1$.
    \end{enumerate}
\end{proposition}

%
%
%

\subsection{Resurgence and the bridge equation}

The results concerning the general formal solutions of equation $E_k$,
together with the general properties of resurgence for holonomic functions at
an irregular singularity with one--level enable us to state the following:

\begin{proposition}
    With the notations of the last proposition, the formal series
    $\widetilde{\Phi_i} (x)$ are resurgent functions for the critical time
    $z = \frac{1}{x}$ ; each has a finite number of singularities in
    the Borel plane, which are regular--singular.
    
    Moreover, the only alien derivations which may act non trivially on $\Phi$ 
    are
    the ${}^{z} \Delta_{u_i - u_j}$ where $0 \leqslant i, j \leqslant k - 1$
    and $i \neq j$ and we have:
    \[ {}^{z} \Delta_{u_i - u_j} (\widetilde{\Phi_i} (x)) = A_{i, j} 
    \widetilde{\Phi_j} (x), \quad A_{i, j} \in \mathbb{C}, \]
    which can be summarized in the following bridge equation
    \[ {}^{z} \pmb{\Delta}_{u_{_i} - u_j} \Phi = \sum A_{i, j} u_j 
    \frac{\partial}{\partial u_j} \Phi. \]
\end{proposition}  

\begin{proof}
    Each $\widetilde{\Phi_i} (x)$ is holonomic and  Gevrey--1 by the formal 
    constructions
    above and we have seen that $\widetilde{\Phi_i} (x)$ has a single
    critical time $z = \frac{1}{x}$.
    
    Accordingly, the Borel transform of
    $\widetilde{\Phi_i}  (x)$ with respect to the critical time $z =
    \frac{1}{x}$ is holonomic and as such it is resurgent, with a finite
    number of singularities which belong to the set of singular points of the
    differential equation in the Borel plane, and exponential growth of order
    $1$ at $\infty$.
    
    The formal integral $\Phi (\lambda)$ is thus resurgent with respect to the
    single critical time $z$. It is then licit to apply to it the operators
    ${}^{z} \Delta_{\omega}$ and the same formal argument used for $(E_2)$
    works:
    
    for any $\omega \in \mathbb{C}$ and any $i \in \{ 0, \ldots, k - 1 \}$,
    $\pmb{\Delta}_{\omega} (\widetilde{\Phi_i} (x))$ is a solution of
    $(E_k)$ and this forces $\omega \in \{ u_i - u_j ; j = 0, \ldots, k - 1; j
    \neq i \}$, with ${}^{z} \Delta_{u_i - u_j} (\widetilde{\Phi_i} (x))
    = \text{constant} . \widetilde{\Phi_j} (x)$.
    
    Finally, the resurgence relations entail that all the singularities of the
    functions $\widehat{\Phi_i}$ are logarithmic and thus of the
    regular--singular type.
\end{proof}

At this stage, we have access to all the information for $\widetilde{\Phi_0}
(x)$ and all the $\widetilde{\Phi_j} (x)$, through the coefficients
of the bridge equation: we have a control on the Stokes phenomenon for these
functions, which for example can be crucial for questions of \emph{real
    resummation}\cite{M97,EM95}.

Let us finally remark that, in the present case, the location of the
singularities could have been directly deduced from the Borel transform of the
operator $H_{k,u_l}$. 
 
\section{Non--linear operations}\label{s:nonlin}

Let us recall the following general result, for resurgent functions\cite{E1,E3}:

\begin{proposition}
    Let $\varphi$ be a resurgent function, with
    $\widetilde{\varphi} (z) \in z^{- 1}\mathbb{C} [[z^{- 1}]]$ and $\chi$ an
    analytic function at the origin; then
    $\chi \circ \varphi$ is resurgent. Moreover, for any $\omega \in
    \mathbb{C}^{\ast}$
    \[ \Delta_{\omega}  (\chi \circ \varphi) = (\partial \chi \circ \varphi)
    \Delta_{\omega} \varphi, \qquad (\text{with } \partial = d / d z). \]
\end{proposition}

\begin{remark}
    This proposition is a particular case of a much more general result on
    composition of resurgent functions, possibly involving in particular 
    Puiseux series
    $\widetilde{\varphi} (z)$. We shall only need the present version and refer 
    to
    the foundational papers by Ecalle for a treatment involving a systematic use
    of majors, for resurgent functions which are non necessarily integrable.
\end{remark}

As $\widetilde{Z_0} (\lambda) - 1 \in \lambda \mathbb{C} [[\lambda]]$, we can
thus state that the free energy $W (\lambda) = \log \widetilde{Z_0}(\lambda)$
is a resurgent function (see Ref.~\onlinecite{KS15} for a recent proof of the 
resurgent character of the logarithm of a resurgent function, under hypotheses 
which are satisfied in our rather elementary situation).
The function $W$ will have an infinite set of singularities in the Borel plane, 
that is the additive semigroup generated by the singularities of 
$\widetilde{Z_0} (\lambda)$ in the Borel plane, yet
only a finite number of alien derivations will act non-trivially on it:
the same ones as for $\widetilde{Z_0} (\lambda)$. There is no contradiction
here, the singularities of the Borel transform of $\widetilde{Z_0} (\lambda)$
are in various sheets of its Riemann surface and they can be reached by
{\emph{compositions}} of alien derivations, which give access to all the
sheets of the Riemann surface of $\widehat{W} (\lambda)$.

For $k = 2$, $G (\lambda) : = W' (\lambda)$ is solution of a Riccati equation,
as it is well known; there are only 2 acting alien derivations $\Delta_{\pm
    u}$ with $u = - 1 / 16$ and the singularities of the Borel transform of $G$
are on $\frac{1}{16} \mathbb{Z}$: there is an infinite number of
singularities but their set is discrete.

For $k > 2$, the situation is not as simple: the function $W (\lambda)$ has
an infinity of singularities which might project on a dense subset in the Borel
plane (see proposition below). The subtle point, however, is that it does not 
prevent $\widehat{W}$ from
being resurgent: in the first sheet (the star of holomorphy), standard results
(see Ref.~\onlinecite{sauzin}, section 5.13)
indeed ensure that it has a finite set of singular directions, with isolated
singularities on these directions and exponential growth at infinity on the
non--singular ones. The other singularities are in the other sheets and
accessible by composition of alien derivations:
\[ \Delta_{\omega_r} \ldots \Delta_{\omega_1} \]
and ultimately, the highly ramified structure of the Riemann surface
$\mathcal{S}$ of $\widehat{W}$ and the behaviour of $W$ when reaching its 
singularities
on $\mathcal{S}$ is totally encoded by the resurgent structure of the
partition function which is as we have seen above quite explicit. 
It is in such a non--linear context that the alien derivations show all their 
efficiency, to explore the surface $\mathcal{S}$ and describe the singularities 
on all its sheets, which eventually govern the analytic properties of the 
function $W$.

Let us end this section with a more precise statement on the sets of 
singularities:
\begin{proposition} For $k>2$ the set of singularities of $W$ is an additive 
subgroup of $\mathbb{C}$ which is discrete if $k\in \lbrace 3,4,5,7 \rbrace$ 
and dense otherwise.
\end{proposition}
\begin{proof}
    Let us recall that the singularities in the Borel plane of $\widetilde{Z_0} 
    (\lambda)$ are the roots $u_1,\dots, u_{k-1}$ of the polynomial :
    $$Q (u) =(- 1)^{k - 1} (2 k)^k (m u)^{k - 1} + 1.$$
    The powers of $\widetilde{Z_0} (\lambda)_1$ in $W(\lambda) = \log 
    \widetilde{Z_0}(\lambda)$ will generate in the Borel plane singularities 
    for $W$ in the set 
    $$\Omega = u_1 \mathbb{N}+\dots u_{k-1}\mathbb{N}$$ which is a priori an 
    additive semigroup in $\mathbb{C}$. This set coincides with the additive 
    subgroup $L$ of $\mathbb{C}$ generated by $u_1,\dots, u_{k-1}$: 
    $\Omega\subset L$ but, since $u_1+\dots+u_{k-1}=0$ any element of $L$ can 
    be written as a linear combination of $u_1,\dots, u_{k-1}$ with 
    coefficients in $\mathbb{N}$.
    
    The roots of $Q$ are invariant by the rotation of angle $\frac{2\pi}{k-1}$ 
    so
    that, whenever $L$ is discrete, it is a lattice in $\mathbb{C}$ invariant by
    this rotation. A classical result on lattices\cite{Coxeter}, sometimes 
    called
    the Crystallographic restriction theorem, ensures that the only possible 
    angles
    are:
    \[ \frac{2\pi}{k-1}=\pi, \frac{2\pi}{3},  \frac{\pi}{2}, \frac{\pi}{3}, \]
    thus $k=3,4,5,7$.
\end{proof}
 
\section{Airy's equation and coequational resurgence.}

The stationary Schr{\"o}dinger equation in one dimension, with a polynomial
potential $W$ is:
\begin{equation}\label{eq:schrod}
\hbar^2 \psi'' (q) - V (q) \psi (q) = 0.
\end{equation}
Around 1980, A. Voros had discovered that the expansions in $\hbar$ of the
solutions of \eqref{eq:schrod} display a resurgent behavior, with a specific 
pattern and
beautiful algebraic properties (``Voros coefficients, Voros algebra''); this
was enhanced by Pham and his school and has lately been the object of numerous
articles (e.g. Ref.~\onlinecite{IWAKI14}).

However, all these works admittedly\cite{VOROS_BOURB83, VOROS83,PHAMDD93}
presupposed
the resurgence of the series -- a fact that for which there was convincing 
numerical evidence. J. Ecalle had given a clear method to establish the 
resurgent character with respect to the critical
time $z = 1 / \hbar$ of the series (which he dubbed ``coequational
resurgence'', to stress on a kind of duality with resurgence properties for
the dynamical variable deduced from $q$), by systematic expansions involving
iterated integrals in the Borel plane yet it was never implemented in details 
for
practical cases.

Using the same ideas as in the previous sections, we show below that, for the
solutions to Airy's equation (which corresponds to $V (q) = q$), the dependence 
in the parameter $\hbar$ is also governed by 
a linear ODE with polynomial coefficients, the variable $q$ being considered 
now as a parameter.

Let us consider Airy's equation 
\begin{equation}
\label{eq:Airy}
\hbar^2 \psi'' (q) - q \psi (q) = 0
\end{equation}
and let $\lambda=\hbar^2$. The solutions to \eqref{eq:Airy} have integral
solutions 
\[
\psi_{\gamma}(q,\lambda)=\int_{\gamma} e^{q\phi-\frac{\lambda}{3}\phi^3}d\phi,
\]
where $\gamma: \mathbb{R} \rightarrow \mathbb{C}$ is an infinite path such
that $\displaystyle \Re (\lambda\gamma(t)^3) 
\underset{|t|\to\infty}{\longrightarrow}+\infty$.
For the same kind of reasons as in theorem \ref{conv}, such an integral defines 
an analytic function in some half plane $H_{\gamma}$ and we have:
\[
\lambda\partial^2_q \psi_{\gamma}(q,\lambda) 
-q\psi_{\gamma}(q,\lambda)=\int_{\gamma} (\lambda \phi^2 
-q)e^{q\phi-\frac{\lambda}{3}\phi^3}d\phi=\left[-e^{q\phi-\frac{\lambda}{3}\phi^3}\right]_{\gamma}=0.
\]

Once $\gamma$ is fixed, we can define analogs of ``moments'' (we view now $q$ 
as a parameter that we omit in the notation):
\begin{equation*}
\forall j \in \mathbb{N},\ \forall \lambda\in H_{\gamma},\quad
Z_j(\lambda)=\int_{\gamma} \phi^j e^{q\phi-\frac{\lambda}{3}\phi^3}d\phi 
\end{equation*}
and, as in proposition \ref{rec}, we get
\begin{equation*}
\forall j \geqslant 0, \quad
\partial_{\lambda}Z_j=Z'_j = - \frac{1}{3} Z_{j +3}
\end{equation*}
and
\begin{equation*}
\forall j \geqslant 0, \quad
j  Z_{j-1} = \lambda Z_{j + 2} -q Z_{j},
\end{equation*}
the latter equations being obtained by integration by parts, with the 
convention $j  Z_{j-1}=0$ for $j=0$ so that $\lambda Z_2=qZ_0$.
Combining the previous equations, with
$\theta_{\lambda}=\lambda \partial_{\lambda}$, yields:
\begin{equation*}
\forall j \geqslant 0, \quad
(3\theta_{\lambda}+j+1)Z_j = 3\lambda Z'_j+(j+1)Z_j
= -\lambda Z_{j + 3}+(j+1)Z_j = -qZ_{j+1}
\end{equation*}
so that we get once again a governing equation for $Z_0$:
\begin{equation}
\label{eq:Gov_Airy}
(3\theta_{\lambda}+2)(3\theta_{\lambda}+1)Z_0
=-q(3\theta_{\lambda}+2)Z_1=q^2Z_2
=\frac{q^3}{\lambda}Z_0.
\end{equation}
As a function of the parameter $\lambda=\hbar^2$, the solutions of Airy's 
equation satisfy a differential equation for which the previous machinery works.
From the formal point of view, assuming that $q\not=0$, we get the Newton
polygon depicted in Fig.~\ref{fig:NP_Airy}.
\begin{figure}
    \begin{center}
        \includegraphics{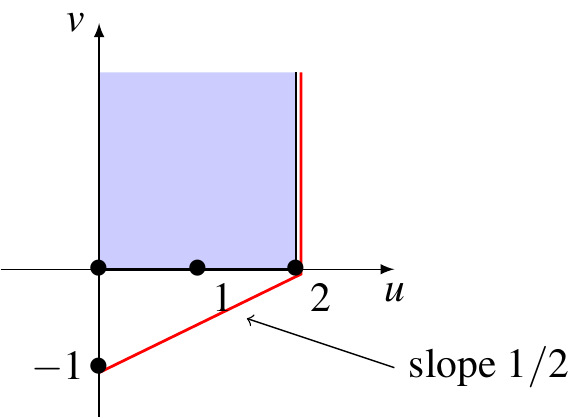}
    \end{center}
    \caption{\label{fig:NP_Airy}
        NP at $0$ for eq.~\eqref{eq:Gov_Airy}, provided $q\neq 0$, with
        $\lambda = \hbar^2$ as variable.
    }
\end{figure}
This suggests the change of variable $x=\lambda^{1/2}$ (that is $x=\hbar$) so 
that the function $f(x)=Z_0(x^2)$ satisfies the equation
\begin{equation*}
\left(\frac{3}{2}\theta_{x}+2\right)\left(\frac{3}{2}\theta_{x}+1\right)f(x)
=\frac{q^3}{x^2}f(x)
\end{equation*}
since $\theta_{\lambda}=\frac{1}{2}\theta_x$. We multiply by 4 and observe that 
now, the NP of the equation 
\begin{equation}\label{eq:airyx}
\left(3\theta_{x}+4\right)\left(3\theta_{x}+2\right)f=\frac{4q^3}{x^2}f(x)
\end{equation}
has a unique slope equal to 1.
We don't get immediately formal solutions but we
perform the change of unknown function $f(x)=e^{u/x}g(x)$ so that
\begin{equation*}
\left(\left(3\theta_{x}+4-3\frac{u}{x}\right)
\left(3\theta_{x}+2-3\frac{u}{x}\right)-\frac{4q^3}{x^2}\right) g(x)=0
\end{equation*}
that also reads
\begin{equation}
\label{eq:A_u}
\left(9 \theta_x^2 +\left(18-\frac{18u}{x}\right)
\theta_x+\left(8-\frac{9u}{x}+\frac{9u^2-4q^3}{x^2}\right)\right)g(x)=0.
\end{equation}
\begin{figure}
    \begin{center}
        \includegraphics{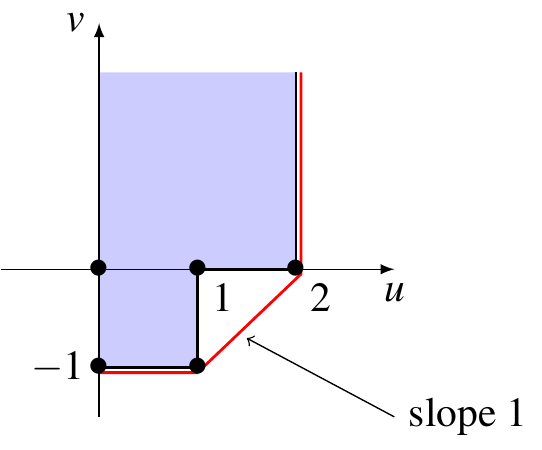}
    \end{center}
    \caption{\label{fig:A_u}
        NP at $0$ for eq.~\eqref{eq:A_u} for $u=u_{\pm}=\pm\frac{2}{3}q^{3/2}$,
        notice its horizontal slope followed
        by its positive slope $1$.
    }
\end{figure}
For $u=u_{\pm}=\pm\frac{2}{3}q^{3/2}$ the NP is given in
Fig.~\ref{fig:A_u}
and the indicial equation gives $\beta=-1/2$ so that the general formal solution
of equation \eqref{eq:airyx} is  
\begin{equation}
f(x)=c_+ x^{-1/2}e^\frac{u_+}{x}h_+(x)+c_- x^{-1/2}e^\frac{u_-}{x}h_-(x)\qquad 
(c_+,c_-\in \mathbb{C})
\end{equation}
where $h_+, h_-$ are formal Gevrey--1 series that are solutions of the 
equations:

\begin{equation*}
\left(9 \theta_x^2  +\left(9-\frac{18u_{\pm}}{x}\right)\theta_x
+\frac{5}{4}\right)h_\pm(x)=0.
\end{equation*}
If $D_x=x^2\partial_x$, $\theta_x=x^{-1} D_x$ so that equation also read:
\begin{align*}
H_\pm.h_\pm 
& = \left(9 \theta_x^2 +\left(9-\frac{18u_{\pm}}{x}\right)\theta_x 
+ \frac{5}{4}\right)h_\pm \\
& = \left( \frac{9}{x^2}D_x^2- \frac{9}{x}D_x
+\frac{1}{x} \left(9-\frac{18u_{\pm}}{x}\right)D_x
+\frac{5}{4}\right) h_\pm \\
& = \left(\frac{9}{x^2}D_x^2 -\frac{18u_{\pm}}{x^2}D_x 
+\frac{5}{4}\right)h_\pm.
\end{align*}
The Borel transform of this equation is 
\begin{align*}
\widehat{H}_\pm \hat{h}_\pm 
& = 9(\zeta^2 \hat{h}_\pm)'' -18u_\pm (\zeta \hat{h}_\pm)''
+\frac{5}{4}\hat{h}_\pm\\
& = (9\zeta^2-18u_\pm\zeta)\hat{h}_\pm''+36(\zeta -u_\pm) \hat{h}_\pm' 
+\frac{77}{4}\hat{h}_\pm.
\end{align*}
The general theory of linear differential equations, together with the NP at 
$\zeta=0$ ensures that the Borel transform of the series $h_\pm$ are analytic 
with $\zeta=2u_\pm$ as unique singularity. The NP at $\zeta=\infty$ also 
ensures that, in any direction avoiding $u_\pm$, the Borel transform has a 
polynomial growth. As in the previous sections, the location of the 
singularities could also be deduced from the associated bridge equation.

In fact, for this equation, all these properties are very well known; the 2 
variables $z$ and $x$ are directly coupled
and the resurgence in $x$ is thus also a consequence of the ``equational
resurgence'' in the dynamical variable $q$, through a rescaling.
For other polynomial
potentials, the challenge
consists in the working out of relevant $x$--dependent integrals in order to 
apply the techniques developed in the present
paper.

\begin{acknowledgments}

The authors wish to thank the ``GDR Renormalisation'' for bringing them together
and Dominique Manchon for organizing the  ``Groupe de travail de Besse'' where 
this article finds its origin.

\end{acknowledgments}


%
\end{document}